\documentclass[journal,twoside,web]{ieeecolor}

\usepackage{generic}
\usepackage{cite}
\usepackage{amsmath,amssymb,amsfonts}
\usepackage{graphicx}
\usepackage{algorithm}
\usepackage{textcomp}
\def\BibTeX{{\rm B\kern-.05em{\sc i\kern-.025em b}\kern-.08em
    T\kern-.1667em\lower.7ex\hbox{E}\kern-.125emX}}

\usepackage[sets,operations,nn,colors,tikz]{cora-macs}

\usepackage{multirow}
\usepackage{dcolumn}
\usepackage{algpseudocode}
\algdef{SE}[DOWHILE]{Do}{doWhile}{\algorithmicdo}[1]{\algorithmicwhile\ #1}%

\usepackage{booktabs}
\usepackage{array} 
\newcolumntype{L}{>{$}l<{$}} 
\newcolumntype{R}{>{$}r<{$}} 
\usepackage{makecell}
\usepackage{tablefootnote}

\let\labelindent\relax
\usepackage[inline]{enumitem}

\usepackage{romannum}

\usepackage{amsthm}
\usepackage{mathrsfs}

\usepackage{aligned-overset}

\usepackage{hyperref}

\usepackage{cleveref}

\newcommand{\X}{\mathcal{X}}
\newcommand{\Y}{\mathcal{Y}}

\newcommand{\Xinit}{\X_\mathrm{I}} 
\newcommand{\Xu}{\X_\mathrm{U}} 
\newcommand{\Xzero}{\X_\mathrm{Z}} 
\newcommand{\numZsets}[0]{\lambda}
\newcommand{\XzeroSub}[1]{{\Xzero}_{,#1}}

\newcommand{\XzeroUnion}[0]{\boldsymbol{\mathcal{X}}_\mathrm{Z}}
\newcommand{\XzeroUnionIter}[1]{\XzeroUnion^{(#1)}}

\newcommand{\BarrierFunc}[0]{\ensuremath{B}}
\newcommand{\nnParams}[0]{\ensuremath{\theta}}
\newcommand{\BarrierNN}[0]{\BarrierFunc_\nnParams}

\newcommand{\Yinit}{\Y_\mathrm{I}} 
\newcommand{\Yu}{\Y_\mathrm{U}} 
\newcommand{\Yzero}{\Y_\mathrm{Z}} 
\newcommand{\YzeroSub}[1]{{\Yzero}_{,#1}}
\newcommand{\YzeroUnion}[0]{\boldsymbol{\mathcal{Y}}_\mathrm{Z}}

\newcommand{\bLower}[1]{\underline{y_\mathrm{#1}}}
\newcommand{\bUpper}[1]{\overline{y_\mathrm{#1}}}
\newcommand{\buLower}[0]{\bLower{U}}
\newcommand{\buUpper}[0]{\bUpper{U}}
\newcommand{\binitLower}[0]{\bLower{I}}
\newcommand{\binitUpper}[0]{\bUpper{I}}

\newcommand{\bzeroLowerSub}[1]{\bLower{Z\ensuremath{,#1}}}
\newcommand{\bzeroUpperSub}[1]{\bUpper{Z\ensuremath{,#1}}}

\newcommand{\shortcZ}[5][]{\shortContSet[#1]{#2,#3,#4,#5}{cZ}} 

\newcommand{\cZ}[3][\contSet{Z}]{#1\vert_{#2\leq#3}}

\newcommand{\setLoss}[0]{\mathcal{L}}
\newcommand{\setLossi}[1]{\setLoss_\mathrm{#1}}
\newcommand{\setLossI}[0]{\setLossi{I}}
\newcommand{\setLossII}[0]{\setLossi{II}}
\newcommand{\setLossIII}[0]{\setLossi{III}}

\newcommand{\lie}{\mathcal{L}}
\newcommand{\zero}{\mathrm{zero}}

\newcommand{\0}[0]{\mathbf{0}}
\newcommand{\1}[0]{\mathbf{1}}

\newcommand{\cmat}[1]{\begin{bmatrix} #1 \end{bmatrix}} 

\newcommand{\zonoConMat}{\ensuremath{C}} 
\newcommand{\zonoConOff}{\ensuremath{d}} 
\newcommand{\numGens}{\ensuremath{q}} 
\newcommand{\abs}[1]{\left|#1\right|}

\newcommand{\ones}{\ensuremath{\mathbf{1}}}

\newcommand{\bigO}{\ensuremath{\mathcal{O}}}

\newcommand{\lightercolor}[3]{
  \colorlet{#3}{#1!#2!white}
}
\lightercolor{CORAcolorBlue}{50}{CORAcolorBlueLight}
\lightercolor{CORAcolorRed}{55}{CORAcolorRedLight}
\lightercolor{CORAcolorYellow}{60}{CORAcolorYellowLight}
\lightercolor{CORAcolorPurple}{70}{CORAcolorPurpleLight}
\lightercolor{CORAcolorGreen}{80}{CORAcolorGreenLight}

\newtheorem{theorem}{Theorem}
\newtheorem{proposition}{Proposition}

\newtheorem{definition}{Definition}

\usepackage[user,lastpage,titleref]{zref}
\usepackage{cleveref}

\crefname{section}{Sec.}{Sec.}
\crefname{subsection}{Sec.}{Sec.}
\crefname{figure}{Fig.}{Fig.}
\crefname{algorithm}{Alg.}{Alg.}
\crefname{table}{Tab.}{Tab.}
\crefname{example}{Ex.}{Ex.}
\crefname{definition}{Def.}{Def.}
\crefname{proposition}{Prop.}{Prop.}
\crefname{corollary}{Cor.}{Cor.}
\crefname{theorem}{Thm.}{Thm.}
\crefname{lemma}{Lemma}{Lemmas}
\crefname{appendix}{Appendix}{Appendix}

\Crefname{section}{Sec.}{Sec.}
\Crefname{subsection}{Sec.}{Sec.}
\Crefname{figure}{Fig.}{Fig.}
\Crefname{algorithm}{Alg.}{Alg.}
\Crefname{table}{Tab.}{Tab.}
\Crefname{example}{Ex.}{Ex.}
\Crefname{definition}{Def.}{Def.}
\Crefname{proposition}{Prop.}{Prop.}
\Crefname{theorem}{Thm.}{Thm.}
\Crefname{corollary}{Cor.}{Cor.}
\Crefname{lemma}{Lemma}{Lemmas}
\Crefname{appendix}{Appendix}{Appendix}

\usepackage{thm-restate}
\AtEndEnvironment{restatable}{
  \begin{proof}
    See \cref{app:proofs}.
  \end{proof}
}

\usepackage{mathtools,stackengine}
\stackMath
\newcommand{\stackEq}[2]{%
  \setbox0=\hbox{${}\mathrel{\stackon[-1pt]{#2}{\scriptstyle #1\strut}}{}$}
  \xdef\tmpwd{\dimexpr\the\wd0\relax}
  \kern.5\tmpwd\mathclap{\box0}&\kern.5\tmpwd
}

\usepackage{pifont}
\newcommand{\firstprop}[0]{P\,\Romannum{1}}
\newcommand{\secondprop}[0]{P\,\Romannum{2}}
\newcommand{\thirdprop}[0]{P\,\Romannum{3}}

\pgfplotsset{
every tick label/.append style={font=\footnotesize},
unsafe set/.style={area legend, thick, draw=CORAcolorUnsafe!85!black, fill=CORAcolorUnsafe!20!white},
initial set/.style={area legend, thick, draw=CORAcolorReachSet!85!black, fill=CORAcolorReachSet!20!white},
invalid zerolevel set/.style={area legend, thick, draw=CORAcolorYellow!85!black, fill=CORAcolorYellow!20!white},
valid zerolevel set/.style={area legend, thick, draw=CORAcolorPurple!85!black, fill=CORAcolorPurple!20!white},
flow arrow/.style={-{Stealth[length=0.75mm, width=0.75mm]}, color=black, point meta={sqrt((\thisrow{u})^2+(\thisrow{v})^2)},quiver={u=\thisrow{u}*0.1, v=\thisrow{v}*0.1,every arrow/.append style={-{Stealth[scale=\pgfplotspointmetatransformed/3000+0.25]}}}},
flow arrow legend/.style={-{Stealth[length=0.75mm, width=0.75mm]}, color=black, quiver={u=\thisrow{u}*0.1, v=\thisrow{v}*0.1,every arrow/.append style={-{Stealth[length=0.75mm, width=0.75mm]}}}},
surface plot/.style={surf, shader=flat, fill opacity=0.2, colormap name=whitegray, mesh/rows=10, draw opacity=0,forget plot},
colormap={mymap}{
    rgb(0pt)=(0.9769,0.9839,0.0805)
    rgb(31pt)=(0.9923,0.8034,0.1939)
    rgb(63pt)=(0.7945,0.7554,0.1570)
    rgb(95pt)=(0.3671,0.8021,0.4563)
    rgb(127pt)=(0.0770,0.7468,0.7224)
    rgb(159pt)=(0.1288,0.6408,0.891)
    rgb(191pt)=(0.2006,0.4803,0.9906)
    rgb(223pt)=(0.2803,0.3367,0.967)
    rgb(255pt)=(0.2422,0.1504,0.6603)
  },
colormap={whitegray}{
    rgb(0pt)=(0.2,0.2,0.2)
    rgb(63pt)=(0.5,0.5,0.5)
    rgb(127pt)=(0.7,0.7,0.7)
    rgb(191pt)=(0.9,0.9,0.9)
    rgb(255pt)=(1.0,1.0,1.0)
  }
}

\begin{document}
\title{Set-Based Training of Neural Barrier Certificates for Safety Verification of Dynamical Systems}

\author{Miriam Kranzlmüller$^*$, Lukas Koller$^*$, Tobias Ladner, and Matthias Althoff
\thanks{$^*$Equal contribution.}
\thanks{Submitted for review on \today. This work was partially supported by the project SPP-2422 (No. 500936349) and the project FAI (No. 286525601), both funded by the German Research Foundation (Deutsche Forschungsgesellschaft, DFG) as well as by the project "Next Generation AI Computing (gAIn)," funded by the Bavarian Ministry of Science and the Arts and the Saxon Ministry for Science, Culture, and Tourism.}
\thanks{All authors are with the Technical University of Munich, Germany. Miriam Kranzlmüller is now with Ludwig-Maximilians-Universität in Munich, Germany (email: miriam.kranzlmueller@math.lmu.de, lukas.koller@tum.de, tobias.ladner@tum.de, althoff@tum.de).}
}

\maketitle

\begin{abstract}
Barrier certificates are scalar functions over the state space of dynamical systems that separate all unsafe states from all reachable states. 
The existence of a barrier certificate formally verifies the safety of the dynamical system.
Recent approaches synthesize barrier certificates by iteratively training a neural network. 
In each iteration, the candidate is formally verified---if successful, the barrier certificate is found.
Instead, we propose a set-based training approach that tightly integrates verification into training via a set-based loss function that soundly encodes all barrier certificate properties.
A loss of zero formally proves the validity of the barrier certificate, collapsing the iterative training and verification into a single training procedure.
Our experiments demonstrate that our set-based training approach scales well with the system dimension and naturally handles complex nonlinear dynamics.
\end{abstract}


\section{Introduction}
\begin{figure*}
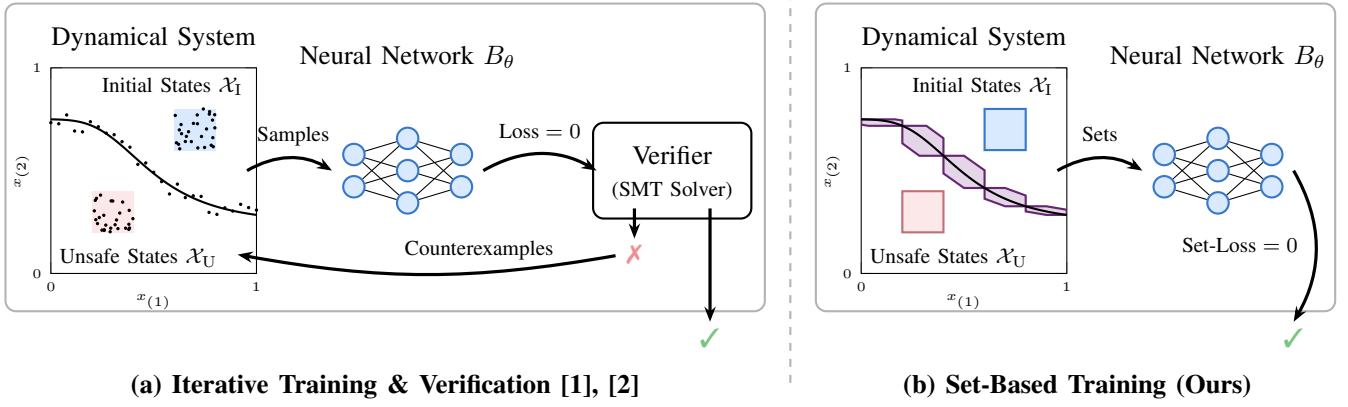

  \centering
  \includetikz[noexport]{./figures/teaser/teaser-figure}
  \caption{Overview: (a) Prior work trains a candidate barrier certificate~$\BarrierNN$ using samples of the state space. The validity of the barrier certificate is verified using a separate verification step. If a counterexample is found, $\BarrierNN$ has to be retrained, resulting in repeated training-verification cycles. (b) In contrast, we use set-based training to integrate the verification directly into the training step, i.e., a loss value of zero directly proves the validity of the barrier certificate.}
  \label{fig:teaser}
\end{figure*}
Autonomous agents operating in safety-critical environments must be formally verified to guarantee safe behavior, e.g., self-driving vehicles~\cite{pek_et_al_2020,althoff_john_2014}, unmanned aerial vehicles~\cite{barry2012}, and robotic manipulators~\cite{garg2024}.
A dynamical system is safe if no trajectory starting from a set of initial states enters a set of unsafe states.
Existing approaches for safety verification of dynamical systems include simulation-based methods~\cite{girard2006}, reachability analysis~\cite{althoff2010}, and barrier certificates~\cite{prajna2006}.
Simulation-based methods have exponential runtime, as the number of required simulations grows exponentially with the system dimension~\cite{girard2006}.
Reachability analysis avoids exponential runtime by enclosing all reachable states using set-based computations.

In this work, we focus on barrier certificates~\cite[Thm.~1]{prajna2004}, which are scalar functions that separate all unsafe states from all reachable states via their zero-level set.
A barrier certificate is characterized by three properties:
\begin{enumerate*}
  \item[\firstprop] unsafe states map to positive values,
  \item[\secondprop] initial states map to non-positive values, and
  \item[\thirdprop] the trajectories at the zero-level set are directed toward the safe region.
\end{enumerate*}
Barrier certificates have been used to prove safety for autonomous and robotics systems~\cite{garg2024}, reinforcement learning~\cite{zhao2022reinforcement}, unmanned aerial vehicles~\cite{barry2012}, and quantum systems~\cite{lewis2023}.

Recent works use an iterative training and verification cycle to synthesize neural networks as barrier certificates~\cite{zhao2020,peruffo2021,zhao2023}, which is visualized in~\cref{fig:teaser}(a).
First, a neural network is trained based on samples of the state space; Second, to ensure safety, a costly verification step is required to verify the properties~\firstprop{} -- \thirdprop{}.
If the properties cannot be verified, the neural network enters a new training and verification cycle.
In contrast, we propose a novel set-based training approach that integrates the verification into the neural network training (\cref{fig:teaser}(b)). 
To this end, the training is augmented using conservative set propagation to compute a sound loss value that directly proves the satisfaction of the required properties, i.e., the safety of the dynamical system is verified once the loss reaches zero. Thereby, we can replace the costly training and verification cycle with a single set-based training procedure.

To summarize, our main contributions are:
\begin{itemize} 
  \item A novel approach to synthesize barrier certificates in a single set-based training procedure, eliminating the training-verification cycle.
  \item A sound set-based loss that encodes the properties of a barrier certificate.
  \item An iterative refinement procedure to dynamically enclose the zero-level set and a sound enclosure of the Lie derivative required for property~\thirdprop{}.
  \item An extensive evaluation on diverse benchmarks with linear, polynomial, and nonpolynomial dynamics. Moreover, we include ablation studies to justify our design choices.
\end{itemize}
\section{Preliminaries}
In this section, we provide the required preliminaries for our set-based training approach.

\subsection{Notation}
Lowercase letters denote vectors, and for a vector $x \in \R^n$, $x_{(i)}$ represents the $i$-th entry.
The vectors of all zeros and all ones are written as $\0$ and $\1$.
The set of all natural numbers up to $n\in\N$ is denoted as $[n] = \{1,\ldots,n\}$.
Matrices are written as uppercase letters.
For a matrix $A \in \R^{n\times m}$, $A_{(i,j)}$ denotes the entry in the $i$-th row and $j$-th column, $A_{(i,\cdot)}$ the $i$-th row and $A_{(\cdot,j)}$ the $j$-th column.
For two matrices $A \in \R^{n\times m_1}$ and $B \in\R^{n \times m_2}$, $[A\;B] \in \R^{n \times (m_1+m_2)}$ describes their horizontal concatenation.
Calligraphic uppercase letters denote sets.
Given a set $\contSet{S}\subseteq\R^n$ and a function $f\colon \R^n\rightarrow\R^m$, the image is written as $f(\contSet{S})\coloneqq \{f(s) \;\vert\; s\in\contSet{S}\}$.
The set of all compact sets over $\R^n$ is denoted by $\mathscr{K}(\R^n)$.

\subsection{Dynamic Systems and Barrier Certificates}
We consider continuous-time dynamical systems that are modeled as ordinary differential equations and operate within a continuous state space $\mathcal{X}\subset\R^n$~\cite[Sec.~2]{peruffo2021}:
\begin{equation}
  \label{eq:cont_sys}
  \dot{x}(t) = f(x),
\end{equation}
where $f\colon \X \to \R^n$ is a continuous vector field.
We want to formally verify the safety of such a dynamical system by
providing a guarantee that, starting from any initial state in $\Xinit$, the system never enters a set of unsafe states $\Xu\subseteq\X$.
To prove this, we require the Lie derivative.
\begin{definition}[Lie Derivative, {\cite[Def. 2.2]{zhao2020}}]
  \label{def:lieDeriv}
  Let $f\colon \X \to \R^n$ be a vector field. The Lie derivative of a continuously differentiable function $\BarrierNN\colon \R^n \to \R$ w.r.t. $f$, $\lie_f \BarrierNN(x)\colon \X \to \R$ is
  \begin{equation*}
    \lie_f \BarrierNN(x) \coloneqq \sum_{i=1}^n \left( \frac{\partial \BarrierNN}{\partial x_{(i)}} (x) \cdot f_{(i)}(x)\right)\text{.} 
  \end{equation*}
\end{definition}

A barrier certificate is a scalar function over the state space, which separates unsafe states and reachable states along its zero-level set.
The following theorem characterizes barrier certificates to prove the safety of continuous-time dynamical systems.
\begin{theorem}[Barrier Certificates, {\cite[Thm.~2.3]{zhao2020}}]
  \label{thm:barrier}
  The system in \eqref{eq:cont_sys} is safe, i.e., there is no trajectory starting in $\Xinit$ that reaches any state in $\mathcal{X}_U$,
  if there exists a differentiable function $B\colon\mathcal{X}\to \mathbb{R}$ satisfying the following conditions:
  \begin{align}
    \label{thm:barrierProp1}
    \forall x \in \Xu\colon    & B(x) > 0\text{,}\tag{\firstprop}        \\
    \label{thm:barrierProp2}
    \forall x \in \Xinit\colon & B(x) \leq 0\text{,}\tag{\secondprop}    \\
    \label{thm:barrierProp3}
    \forall x\in \Xzero\colon  & \lie_f B(x) < 0\text{,}\tag{\thirdprop}
  \end{align}
  where $\Xzero \coloneqq \{ x \in \X\mid B(x) = 0\}$ is the zero-level set.
\end{theorem}

\subsection{Set-Based Computing}
As mentioned previously, we utilize set-based computations to simultaneously train and verify a neural network.
The simplest set representation we use is an $n$-dimensional interval $\I \subset \R^n$ with bounds $l,u \in \R^n$, denoted by $\I \coloneqq \left[ l, u \right]$ with $\forall i \in [n]\colon l_{(i)} \leq x_{(i)} \leq u_{(i)}$.
We also require zonotopes, which are projections of a higher dimensional unit hypercube.
\begin{definition}[Zonotope {\cite[Def. 7.13]{ziegler1995}}]
  \label{def:zonotope}
  Given a center $c \in \R^n$ and a generator matrix $G \in \R^{n\times q}$, a zonotope $\Z \subset \R^n$ is defined as
  \begin{equation*}
    \Z \coloneqq \shortZ{c}{G}= \left\{ c+G \beta \mid \beta  \in [-1,1]^q \right\}\text{.}
  \end{equation*}
\end{definition}
Given a zonotope $\Z_1 = \shortZ{c_1}{G_1} \subset \R^n$, a matrix $W \in \R^{m\times n}$ and a vector $b \in \R^m$, the affine map is computed as~\cite[Lem.~2]{kuehn1998}
\begin{equation}
  \label{eq:affineMap}
  W \Z_1 +b = \{W z_1 + b \mid z_1 \in \Z_1 \} = \shortZ{Wc_1+b}{WG_1}\text{.}
\end{equation}
The Minkowski sum of $\Z_1$ and zonotope $\Z_2 = \shortZ{c_2}{G_2} \subset \R^n$ is computed as~\cite[Lem.~2]{kuehn1998}
\begin{align}
  \begin{split}
    \Z_1 \oplus \Z_2 &\coloneqq \{z_1+z_2 \mid z_1 \in \Z_1,\,z_2 \in \Z_2\} \\ & = \shortZ{c_1+c_2}{[G_1\;G_2]}\subset \R^n.
  \end{split}
\end{align}
We enclose the inner product of zonotopes as~\cite[Thm. 1]{althoff2013}
\begin{align}
  \label{eq:zonoproduct}
  \begin{split}
    \Z_1^\top \cdot \Z_2 & \coloneqq\{x_1^\top x_2 \mid x_1\in\Z_1, x_2\in \Z_2\} \subseteq \shortZ{c'}{G'}\text{,} \\
    c' & = c_1^\top c_2\text{,} \\ G' & = [c_1^\top G_2,\;G_1^\top c_2,\;G_{1(\cdot,1)}^\top G_2,\;\ldots\;G_{1(\cdot,q_1)}^\top G_2]\text{.}
  \end{split}
\end{align}

The factors~$\beta$ of a zonotope~(\cref{def:zonotope}) can be constrained to represent arbitrary convex polytopes.
\begin{definition} [Constrained Zonotope, {\cite[Def. 3]{scott2016}}]
  \label{def:conZono}
  Given a zonotope $\Z=\shortZ{c}{G}$ with $c \in \R^n$ and $G \in \R^{n\times q}$, and $A \in \R^{p \times q}$, $b \in \R^p$, a constrained zonotope\footnote{In contrast to \cite[Def. 3]{scott2016}, we define the constraint via an inequality for convenience.} is defined as
  \begin{align*}
    \cZ{A}{b} & \coloneqq \shortcZ{c}{G}{A}{b}                                                \\
              & = \left\{ c+G \beta \mid \beta  \in [-1,1]^q,\,A\beta \leq b \right\}\text{.}
  \end{align*}
\end{definition}

\subsection{Neural Networks and Set-Based Training}\label{sec:nntraining}

A feed-forward neural network $\NN\colon\R^{\numNeurons_0} \to \R^{\numNeurons_\numLayers}$ is a list of $\numLayers \in \N$ alternating linear and nonlinear layers.
\begin{definition}[Neural Network, {\cite[Sec.~5.1]{bishop2006}}]\label{def:neural_network}
  For an input $\nnInput\in\R^{\numNeurons_0}$, the output of each layer $\nnHidden_k$, for $k\in[\numLayers]$, and the output of a neural network $\nnOutput = \NN(\nnInput)\in\R^{\numNeurons_\numLayers}$ are
  \begin{align*}
    \nnHidden_0 & = \nnInput\text{,}                                               \\
    \nnHidden_k & = \begin{cases}
                      W_k\,\nnHidden_{k-1} + b_k & \text{if $k$-th layer is linear,} \\
                      \actfun_k(\nnHidden_{k-1}) & \text{otherwise,}
                    \end{cases} \\
    \nnOutput   & = \nnHidden_\numLayers\text{,}
  \end{align*}
  with weight matrices~$W_k \in \R^{\numNeurons_k \times \numNeurons_{k-1}}$, bias vectors~$b_k \in \R^{\numNeurons_k}$, and element-wise activation functions~$\nnActFun_k$. The parameter vector~$\nnParams$ of the neural network contains all weights and biases.
\end{definition}

\subsubsection{Set-Based Forward Propagation}
To compute the set-based loss function required for our approach, we conservatively propagate zonotopes through all layers of the neural network.
The image of a linear layer is computed by applying an affine map~\eqref{eq:affineMap}.
The image of nonlinear layers must be enclosed, as zonotopes are not closed under nonlinear functions.
\begin{proposition}
  [Set-Based Forward Prop., {\cite[Sec. 2.3]{ladner2023}}]
  \label{def:setbasedForwardProp}
  Given a neural network $\NN\colon\R^{\numNeurons_0} \to \R^{\numNeurons_\numLayers}$ and an input set $\nnInputSet \subset \R^{\numNeurons_0}$, an enclosure $\nnOutputSet\subset\R^{\numNeurons_\numLayers}$ of the output set $\nnOutputSetExact = \NN(\nnInputSet)$, i.e., $\nnOutputSetExact\subseteq\nnOutputSet$, is computed as
  \begin{align*}
    \nnHiddenSet_0 & = \nnInputSet,                                                                         \\
    \nnHiddenSet_k & =  \begin{cases}
                          W_k \nnHiddenSet_{k-1} + b_k                & \text{if $k$-th layer is linear,} \\
                          \opEnclose{\nnActFun_k}{\nnHiddenSet_{k-1}} & \text{otherwise,}
                        \end{cases} \\
    \nnOutputSet   & = \nnHiddenSet_\numLayers\text{,}
  \end{align*}
  where $k\in[\numLayers]$ and $\opEnclose{\nnActFun_k}{\nnHiddenSet_{k-1}}\supseteq\nnActFun_k(\nnHiddenSet_{k-1})$ encloses the output set of the $k$-th nonlinear layer. 

\end{proposition}
We write $\opEnclose{\NN}{\nnInputSet}=\nnOutputSet$ to denote the set propagation through the entire neural network.
\subsubsection{Set-Based Training}
Standard loss functions cannot provide safety guarantees, because they evaluate the neural network at finitely many output points.
Therefore, we use set-based training, which generalizes standard neural network training by integrating the set propagation into the training process~\cite{koller2025}.
Using the set propagation during training, we can define set-based loss functions, i.e., loss functions defined for entire output sets.
Intuitively, the goal of set-based neural network training is to minimize a set-based loss function $\setLoss \colon \mathscr{K}(\R^{\numNeurons_\numLayers})\times\R^{\numNeurons_\numLayers} \to \R$ over continuous output sets of the neural network:
\begin{equation}\label{eq:set_based_training_objective}
  \min_\nnParams \sum_{\substack{(\nnInputSet,\nnTarget)\sim\mathcal{D} \\ \nnOutputSet=\opEnclose{\NN_\nnParams}{\nnInputSet}}}\setLoss(\nnOutputSet,\nnTarget)\text{,}
\end{equation}
where $\mathcal{D}$ is a dataset containing input sets~$\nnInputSet\subset\R^{\numNeurons_0}$ and targets~$\nnTarget\in\R^{\numNeurons_\numLayers}$.

\subsection{Problem Statement}
Given a continuous-time dynamical system $\dot{x}(t) = f(x)$ with state space $\mathcal{X}\subset \mathbb{R}^n$, an initial set $\mathcal{X}_I \subseteq \mathcal{X}$, and an unsafe set $\mathcal{X}_U \subseteq \mathcal{X}$, we want to formally prove safety by synthesizing a barrier certificate~$B\colon \R^n\rightarrow \R$ using set-based training.
\section{Set-based Training of Barrier Certificates}
\label{sec:set-basedBC}
\begin{figure*}
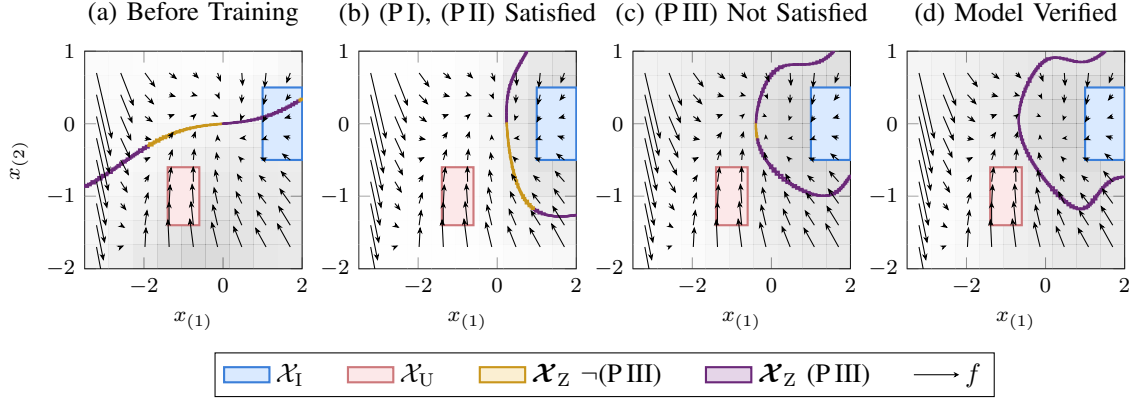

  \centering
  \includetikz[export]{./figures/experiments-groupplot-evolution/evolution-groupplot}
  \caption{Visualization of the neural barrier certificate at different stages of the training process.}
  \label{fig:barrierlearning}
\end{figure*}
Our novel approach uses set-based training to synthesize barrier certificates without requiring a subsequent verification step. The set-based loss function directly proves the validity of the barrier certificate.

\subsection{Sound Set-Based Loss Function}
\label{sec:setBasedLoss}

In this subsection, we define a set-based loss function that can verify the validity of a barrier certificate.
For each \firstprop{} -- \thirdprop{}, we define a set-based loss function that can verify the corresponding property.

\firstprop{} requires all unsafe states $\Xu$ to be mapped to positive values by $\BarrierNN$.
We enclose the image of $\Xu$ using~\cref{def:setbasedForwardProp},
\begin{align}
  \Yu = \opEnclose{\BarrierNN}{\Xu} \supseteq \BarrierNN(\Xu)\text{,}
\end{align}
and the loss function penalizes the lower bound $\buLower$ of $\Yu$, i.e., for $\Yu\subseteq\shortI{\buLower}{\buUpper}$ and some small $\epsilon > 0$:
\begin{equation}
  \label{eq:lossI}
  \setLossI(\Yu) \coloneqq \max(0,-\buLower + \epsilon)\text{.}
\end{equation}
Intuitively, if $\setLossI(\Yu) = 0$, then $-\buLower + \epsilon \leq 0 \implies \buLower > 0$ and hence \firstprop{} is satisfied.
We require a small $\epsilon > 0$ to ensure a strict inequality.

\secondprop{} requires all initial states to be mapped to non-positive values.
Hence, we define the set-based loss analogously to $\setLossI$, i.e., for $\Yinit = \opEnclose{\BarrierNN}{\Xinit} \supseteq \BarrierNN(\Xinit)$ with bounds $\shortI{\binitLower}{\binitUpper}\supseteq\Yinit$, the loss is defined as
\begin{equation}
  \label{eq:lossII}
  \setLossII(\Yinit) \coloneqq \max(0,\binitUpper)\text{.}
\end{equation}
Again, if $\setLossII(\Yinit) = 0$, then $\binitUpper \leq 0$ and hence \secondprop{} is satisfied.

\thirdprop{} requires the Lie derivative of all zero-level points to be negative. The zero-level set $\Xzero$ is not known given $\X$ and $\BarrierNN$, and changes after each training loop as the parameters $\nnParams$ are updated.
Thus, we compute a zero-level set enclosure $\XzeroUnion$ from $\X$ in each training iteration; the details are deferred to~\cref{sec:zerolevelset}.

On a high-level, we enclose the zero-level set using a union of zonotopes:
\begin{equation}
  \bigcup_{i=1}^\numZsets \XzeroSub{i} \supseteq \Xzero\text{,}
\end{equation}
which we write as follows for convenience:
\begin{equation}
  \label{eq:Xzero-union}
  \XzeroUnion \coloneqq \{\XzeroSub{1},\ldots,\XzeroSub{\numZsets}\}.
\end{equation}
For each $\XzeroSub{i}\in\XzeroUnion$, we enclose the Lie derivative:
\begin{equation}
  \label{eq:enclosureOfLieDeriv}
  \YzeroSub{i} = \opEnclose{\lie_f\BarrierNN}{\XzeroSub{i}} \supseteq \lie_f \BarrierNN(\XzeroSub{i})\text{.}
\end{equation}
Analogously, we collect these enclosures into a set:
\begin{equation}
  \label{eq:Yzero-union}
  \YzeroUnion \coloneqq \{\YzeroSub{1},\ldots,\YzeroSub{\numZsets}\}.
\end{equation}
Finally, we penalize the upper bound of each set, i.e., for $\YzeroSub{i} = \shortI{\bzeroLowerSub{i}}{\bzeroUpperSub{i}}$ and some small $\epsilon > 0$:

\begin{equation}
  \label{eq:lossIII}
  \setLossIII(\YzeroUnion) \coloneqq \sum_{i=1}^\numZsets\max(0,\bzeroUpperSub{i} + \epsilon)\text{.}
\end{equation}
Hence, if $\setLossIII(\YzeroUnion) = 0$, then for all $i\in[\numZsets]$: $\bzeroUpperSub{i}+\epsilon\leq 0$, i.e., $\bzeroUpperSub{i}\leq-\epsilon<0$, and hence \thirdprop{} is satisfied.
We require $\epsilon>0$ to enforce the strict inequality required by~\thirdprop{}.

Ultimately, the total set-based loss function sums the losses for three properties:
\begin{equation}
  \label{eq:total-loss}
  \setLoss(\Yu,\Yinit,\YzeroUnion) \coloneqq \setLossI(\Yu) + \setLossII(\Yinit) + \setLossIII(\YzeroUnion)\text{.}
\end{equation}
The set-based loss function~$\setLoss(\Yu,\Yinit,\YzeroUnion)$ encodes all three characteristic properties of a barrier certificate~(\cref{thm:barrier}). Therefore, if $\setLoss(\Yu,\Yinit,\YzeroUnion) = 0$, the safety of the dynamical system is verified.
\begin{restatable}[Soundness of Set-Based Loss]{theorem}{thmsoundnessloss}
  \label[theorem]{thm:soundness-loss}
  The system in \eqref{eq:cont_sys} is safe if
  \begin{equation*}
    \setLoss(\Yu,\Yinit,\YzeroUnion) = 0\text{.}
  \end{equation*}
\end{restatable}

\cref{fig:barrierlearning} illustrates the evolution of the barrier certificate throughout training.
Initially, none of the properties~\firstprop{} -- \thirdprop{} are satisfied~(\cref{fig:barrierlearning}a).
As training progresses, \firstprop{} and \secondprop{} are satisfied first, while parts of $\XzeroUnion$ still violate~\thirdprop{}~(\cref{fig:barrierlearning}b).
These violations are successively eliminated~(\cref{fig:barrierlearning}c) until \thirdprop{} holds for all sets in $\XzeroUnion$, formally verifying the system~(\cref{fig:barrierlearning}d).

Subsequently, we describe the details required for computing~$\setLossIII$, which requires enclosing the zero-level set~(\cref{sec:zerolevelset}) and enclosing the Lie derivative~(\cref{sec:lie-derivative}).

\subsection{Efficient Enclosure of the Zero-Level Set}
\label{sec:zerolevelset}

The zero-level set of a neural network is generally nonconvex and can be arbitrarily complex, as illustrated in~\cref{fig:barrierlearning}.
We therefore represent it as a union of zonotopes, computed via a recursive branch-and-bound strategy:
we split the input space and discard any region that provably contains no zero point.
To do so, we enclose the preimage of $\BarrierNN$ for the output value $0$ by exploiting the shared factor space between the input and output zonotopes~\cite{koller2025shadows}. The following proposition formalizes the preimage enclosure.
\begin{proposition}[Enclosing the Zero-Level Set, {\cite[Prop.~2]{koller2025shadows}}]\label{prop:input_set_refinement}
  Given are a neural network $\BarrierNN\colon\X\to\R$ and an input set~$\XzeroSub{i}=\shortZ{c_x}{G_x}\subset\R^n$, where $G_x\in\R^{n\times \numGens_0}$. We can enclose the zero-level set of $\BarrierNN$ by

  \begin{equation*}
    \BarrierNN^{-1}(\BarrierNN(\XzeroSub{i}) \cap \{0\}) \subseteq \cZ[\XzeroSub{i}]{\zonoConMat}{\zonoConOff}\text{,}
  \end{equation*}
  where $\nnOutputSet=\shortZ{c_y}{G_y}=\opEnclose{\BarrierNN}{\XzeroSub{i}}$ is the output enclosure and the constraints are defined as
  \begin{align*}
    \zonoConMat & \coloneqq \cmat{-G_{y(\cdot,[\numGens_0])} \\ G_{y(\cdot,[\numGens_0])}}\text{,} \\
    \zonoConOff & \coloneqq \cmat{c_y                        \\ -c_y} + \abs{\cmat{G_{y(\cdot,[\numGens_\numLayers]\setminus [\numGens_0])} \\ G_{y(\cdot,[\numGens_\numLayers]\setminus [\numGens_0])}}}\,\ones\text{.}
  \end{align*}
\end{proposition}

\cref{fig:zerolevelsets} illustrates the zero-level set enclosure using an example: starting from the full state space~$\X$, four iterations of preimage enclosure and splitting yield a tight union of zonotopes enclosing the nonconvex zero-level set.

\begin{figure}
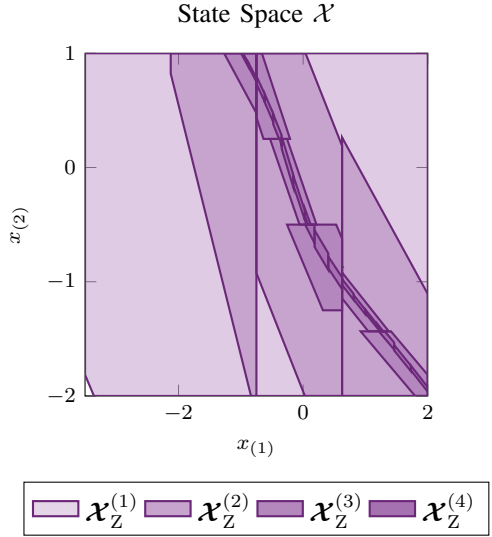

  \centering
  \includetikz[export]{./figures/zerolevelsets}
  \caption{We enclose the zero-level set by recursively splitting the state space and enclosing the preimage enclosure to discard input regions that do not contain a zero point. After four iterations we achieve a tight enclosure of the zero-level set.}
  \label{fig:zerolevelsets}
\end{figure}

\subsection{Enclosure of the Lie Derivative}
\label{sec:lie-derivative}

For the set-based loss of~\thirdprop{}, we enclose the Lie derivative for each~$\XzeroSub{i}\in\XzeroUnion$.
The Lie derivative~(\cref{def:lieDeriv}) is the inner product of the gradient of~$\BarrierNN$ and the flow~$f$, and we over-approximate it by first computing the set of gradients followed by enclosing the flow $f$.
\begin{restatable}[Enclosing Sets of Gradients]{proposition}{propenclosingsetbackprop}
  \label{prop:enclosing-setbackprop}
  Given a neural network $\BarrierNN\colon\X\to\R$ and an input set~$\XzeroSub{i}\subset\R^n$, an enclosure of the set of input gradients~$\nnGradSet_{\XzeroSub{i}}\supseteq\{\nabla_{\nnInput}\BarrierNN(\nnInput)\mid \nnInput\in\XzeroSub{i}\}$ is computed for $k=\numLayers,\dots,1$
  \begin{align*}
    \nnGradSet_\numLayers & = \shortZ{1}{\0},                                                                                       \\
    \nnGradSet_{k-1}      & = \begin{cases}
                                W_{k}^\top \nnGradSet_k                                       & \text{if $k$-th layer is linear,} \\
                                \opEnclose{\nnActFun_k'}{\nnHiddenSet_{k-1}}\cdot\nnGradSet_k & \text{otherwise.}
                              \end{cases}
  \end{align*}
  Finally, $\nnGradSet_{\XzeroSub{i}}\coloneqq \nnGradSet_{0}\supseteq\{\nabla_{\nnInput}\BarrierNN(\nnInput)\mid \nnInput\in\XzeroSub{i}\}$.
\end{restatable}
We enclose the flow $f$ for each zero-level set $\XzeroSub{i}$, i.e., $\opEnclose{f}{\XzeroSub{i}}\supseteq\{f(\nnInput)\mid \nnInput\in\XzeroSub{i}\}$: For linear systems, we can compute the flow exactly using an affine map~\eqref{eq:affineMap}; other dynamics have to be enclosed, e.g., using range bounding based on interval arithmetic~\cite[Sec. 2.3.3]{jaulin2001} or linearization~\cite{althoff_2008}.
Finally, we enclose the Lie derivative using~\eqref{eq:zonoproduct}:
\begin{equation}
  \label{eq:lie-enclosure}
  \lie_f\BarrierNN(\XzeroSub{i}) \subseteq \nnGradSet_{\XzeroSub{i}}^\top\opEnclose{f}{\XzeroSub{i}}\text{.}
\end{equation}

\subsection{Algorithm}

\cref{alg:learn-nbc} realizes our set-based training for neural barrier certificates:
First, the neural network~$\BarrierNN$ is initialized.
Second, the losses for all unsafe states~\firstprop{} and all initial states~\secondprop{} are computed.
Third, the zero-level set is enclosed.
Fourth, for each zero-level set, the gradient set and the Lie derivative are enclosed.
Fifth, the losses for all zero-level sets are computed.
Finally, the total loss is computed based on which the parameters of the neural network are updated.
Once the total loss is zero, the neural network is a valid barrier certificate~(\cref{thm:soundness-loss}) and the training terminates.

\begin{algorithm}
  \caption{Set-Based Training of Neural Barrier Certificates}
  \label{alg:learn-nbc}
  \begin{algorithmic}[1]
    \Require State space $\X \subset \R^n$, initial set $\Xinit \subseteq \X$, unsafe region $\Xu \subseteq \X$, and dynamics $f\colon \X \to \R^n$.
    \State Initialize network $\BarrierNN\colon \X \to \R$
    \Do
    \State \texttt{--- Property~\firstprop{} ---} \Comment{\cref{sec:setBasedLoss}}
    \State Enclose~$\BarrierNN$ for all unsafe states~$\Xu$ \Comment{\cref{def:setbasedForwardProp}}
    \State $\setLossI\gets$ Compute loss for~\firstprop{} \Comment{\eqref{eq:lossI}}
    \State \texttt{--- Property~\secondprop{} ---}
    \State Enclose~$\BarrierNN$ for all initial states~$\Xinit$ \Comment{\cref{def:setbasedForwardProp}}
    \State $\setLossII\gets$ Compute loss for~\secondprop{} \Comment{\eqref{eq:lossII}}
    \State \texttt{--- Property~\thirdprop{} ---}
    \State $\XzeroUnion\gets$ Enclose zero-level set \Comment{\cref{sec:zerolevelset}}
    \For{each $\XzeroSub{i}\in\XzeroUnion$}
    \State Enclose gradient set for~$\XzeroSub{i}$\Comment{\cref{sec:lie-derivative}}
    \State Enclose the Lie derivative for~$\XzeroSub{i}$ \Comment{\cref{sec:lie-derivative}}
    \State ${\setLossIII}_{,i}\gets$ Compute loss for~$\XzeroSub{i}$
    \State $\setLossIII\gets \setLossIII + {\setLossIII}_{,i}$ \Comment{\eqref{eq:lossIII}}
    \EndFor
    \State \texttt{--- Update step ---}
    \State $\setLoss\gets \setLossI + \setLossII + \setLossIII$ \Comment{Compute total loss~\eqref{eq:total-loss}}
    \State $\nnParams\gets\nnParams - \eta\nabla_\nnParams \setLoss$ \Comment{Update the parameters}
    \doWhile{$\setLoss > 0$}
    \State \Return valid barrier certificate $\BarrierNN$ \Comment{\cref{thm:soundness-loss}}
  \end{algorithmic}
\end{algorithm}
\section{Experimental Results}
\begin{table*}
  \centering
  \footnotesize
  \caption{Main results: Comparing runtime and success rate. The results are averaged over 10 random neural network initializations.}
  \label{tab:main-results}
  \begin{tabular}{ l l c R@{$\pm$}L R<{\%} R@{$\pm$}L R<{\%} R@{$\pm$}L R<{\%} R@{$\pm$}L R<{\%}}
    \toprule
                              &               &     & \multicolumn{3}{c}{\normalsize Ours} & \multicolumn{3}{c}{\normalsize SynNBC} & \multicolumn{3}{c}{\normalsize FOSSIL} & \multicolumn{3}{c}{\normalsize PRoTECT~}                                                                                                                                                                            \\
    \cmidrule(lr){4-6}  \cmidrule(lr){7-9}  \cmidrule(lr){10-12} \cmidrule(lr){13-15}
    Dynamics                  & Benchmark     & $n$ & \multicolumn{2}{c}{Time [s]}         & \multicolumn{1}{c}{Success}            & \multicolumn{2}{c}{Time [s]}           & \multicolumn{1}{c}{Success}              & \multicolumn{2}{c}{Time [s]} & \multicolumn{1}{c}{Success} & \multicolumn{2}{c}{Time [s]} & \multicolumn{1}{c}{Success}                                                  \\
    \midrule
    \multirow{5}{*}{Linear}   & Three Sets    & 2   & 0.5                                  & 0.16                                   & 100                                    & \multicolumn{2}{c}{--}                   & 0                            & \bf{0.1}                    & 0.09                         & 100                         & 0.2                    & 0.01     & 100        \\
                              & Two Barriers  & 2   & \bf{0.6}                             & 0.23                                   & 100                                    & \multicolumn{2}{c}{--}                   & 0                            & 2.8                         & 1.90                         & 100                         & \multicolumn{2}{c}{--} & 0                     \\
                              & Peruffo (4D)  & 4   & \bf{0.5}                             & 0.02                                   & 100                                    & 11.6                                     & 0.34                         & 100                         & 9.9                          & 8.66                        & 100                    & 0.8      & 0.01 & 100 \\
                              & Peruffo (6D)  & 6   & \bf{3.1}                             & 2.50                                   & 100                                    & 6.3                                      & 0.10                         & 100                         & 12.2                         & 12.05                       & 70                     & 7.7      & 0.02 & 100 \\
                              & Peruffo (8D)  & 8   & \bf{1.0}                             & 1.27                                   & 100                                    & 25.6                                     & 4.19                         & 100                         & 47.7                         & 20.04                       & 80                     & 67.8     & 0.44 & 100 \\
    \midrule
    \multirow{3}{*}{Polynomial}    & Darboux       & 2   & \bf{5.2}                             & 3.44                                   & 100                                    & \multicolumn{2}{c}{--}                   & 0                            & 5.7                         & 5.19                         & 90                          & \multicolumn{2}{c}{--} & 0                     \\
                              & Polynomial    & 2   & 1.4                                  & 0.24                                   & 100                                    & 0.7                                      & 0.02                         & 100                         & 1.0                          & 0.89                        & 100                    & \bf{0.3} & 0.01 & 100 \\
                              & Lyapunov      & 3   & 2.1                                  & 1.40                                   & 100                                    & 1.2                                      & 1.00                         & 100                         & \bf{0.5}                     & 0.48                        & 100                    & 16.9     & 0.22 & 100 \\
    \midrule
    \multirow{5}{*}{Nonpolynomial} & Exponential   & 2   & 19.0                                 & 13.29                                  & 100                                    & \multicolumn{2}{c}{--}                   & 0                            & \bf{0.1}                    & 0.10                         & 100                         & \multicolumn{2}{c}{--} & 0                     \\
                              & Ratschan (3D) & 3   & \bf{0.5}                             & 0.02                                   & 100                                    & \multicolumn{2}{c}{--}                   & 0                            & 1.0                         & 0.10                         & 100                         & \multicolumn{2}{c}{--} & 0                     \\
                              & Ratschan (5D) & 5   & \bf{0.5}                             & 0.03                                   & 100                                    & \multicolumn{2}{c}{--}                   & 0                            & 0.9                         & 0.09                         & 100                         & \multicolumn{2}{c}{--} & 0                     \\
                              & Ratschan (7D) & 7   & \bf{0.6}                             & 0.34                                   & 100                                    & \multicolumn{2}{c}{--}                   & 0                            & 1.0                         & 0.11                         & 100                         & \multicolumn{2}{c}{--} & 0                     \\
                              & Ratschan (9D) & 9   & \bf{1.1}                             & 0.98                                   & 100                                    & \multicolumn{2}{c}{--}                   & 0                            & 1.6                         & 1.13                         & 100                         & \multicolumn{2}{c}{--} & 0                     \\
    \bottomrule
  \end{tabular}
\end{table*}
We have implemented our set-based training of neural barrier certificates using the MATLAB Toolbox CORA~\cite{althoff2015,Althoff2025manual}, and the code will be made available with the next release.
\cref{app:details-experiment-benchmark} contains details on the experiments and benchmarks.

\begin{figure*}
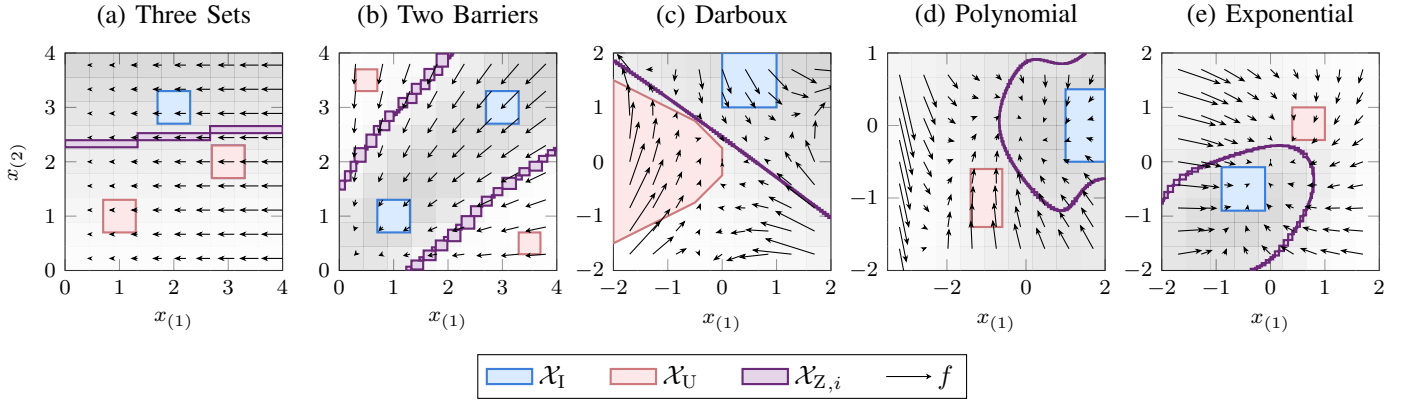

  \centering
  \includetikz[export]{./figures/experiments-groupplot/groupplot}
  \caption{Visualization of computed barrier certificates for different benchmarks.}
  \label{fig:barriervisualization}
\end{figure*}

\subsection{Main Results}
We evaluate our approach on 13 benchmarks spanning linear, polynomial, and nonpolynomial dynamics~(\cref{app:details-experiment-benchmark}): Peruffo~(4D, 6D, 8D)~\cite{peruffo2021}, Darboux~\cite{zeng2016}, Polynomial~\cite{prajna2007}, Lyapunov~\cite{ratschan2006}, Exponential~\cite{liu2015}, Ratschan~(3D, 5D, 7D, 9D)~\cite{ratschan2017}, and two custom benchmarks with multiple initial sets and unsafe regions (Three Sets, Two Barriers).
Because neural networks are randomly initialized~\cite{glorot2010}, each experiment is repeated for ten different random seeds.
We report the \emph{success rate} (percentage of trainings resulting in a valid barrier certificate) and the mean verification \emph{time} with standard deviation.

We compare against three state-of-the-art tools:
\begin{enumerate*}[label=(\roman*)]
  \item SynNBC~\cite{zhao2023},
  \item FOSSIL~\cite{fossil1,fossil2}, both of which use an iterative training-verification cycle (\cref{fig:teaser}a), and
  \item PRoTECT~\cite{Wooding2025}, which uses sum-of-squares optimization to construct polynomial barrier certificates.
\end{enumerate*}

\cref{tab:main-results} reports the results and \cref{fig:barriervisualization} depicts illustrations of valid barrier certificates.
Our approach achieves a 100\% success rate on all considered benchmarks~(\cref{tab:main-results}), demonstrating robustness across linear, polynomial, and nonpolynomial dynamics, as well as benchmarks with multiple initial and unsafe sets.

In contrast, SynNBC and PRoTECT are fundamentally limited to polynomial dynamics and single-set configurations (indicated by~``--'' in \cref{tab:main-results}).
Consequently, they cannot handle the nonpolynomial benchmarks (Exponential, Ratschan) or the two benchmarks with multiple sets (Three Sets, Two Barriers).
Furthermore, SynNBC and PRoTECT rely on intervals; the parabola-shaped unsafe region of the Darboux benchmark cannot be enclosed tightly by an interval, so both tools fail, whereas our approach succeeds by using a zonotopic enclosure.

FOSSIL supports all benchmark categories and also achieves 100\% success on most benchmarks, but is consistently slower than our approach: on Peruffo~(8D), FOSSIL requires on average $47.7$\,s compared to our $1.0$\,s, and our approach is faster on more than half of all benchmarks.
Moreover, FOSSIL's success rate degrades on higher-dimensional systems, e.g., only $70\%$ and $80\%$ on Peruffo~(6D) and Peruffo~(8D), respectively.

In summary, our approach is the only method that solves all considered benchmarks at 100\% success rate, handles more expressive dynamics and set representations than SynNBC and PRoTECT, and is faster than FOSSIL.

\subsection{Training Epochs vs. Training-Verification Cycles}

\begin{table}
  \centering
  \caption{Number of training epochs of our approach and number of training-verification cycles of related approaches.}\label{tab:training_epochs}
  \footnotesize
  \begin{tabular}{l c R@{$\pm$}L R@{$\pm$}L R@{$\pm$}L}
    \toprule
                  &     & \multicolumn{2}{c}{\#Training Epochs} & \multicolumn{4}{c}{\#Training-Verification Cycles}                                                    \\
    \cmidrule(lr){3-4} \cmidrule(lr){5-8}
    Benchmark     & $n$ & \multicolumn{2}{c}{Ours}              & \multicolumn{2}{c}{SynNBC}                         & \multicolumn{2}{c}{FOSSIL}                       \\
    \midrule
    Three Sets    & 2   & 16.1                                  & 10.85                                              & \multicolumn{2}{c}{--}     & 1.10 & 0.32         \\
    Two Barriers  & 2   & 27.7                                  & 34.05                                              & \multicolumn{2}{c}{--}     & 3.6  & 2.50         \\
    Peruffo (4D)  & 4   & 25.1                                  & 5.32                                               & 5.0                        & 0.00 & 4.9  & 3.48  \\
    Peruffo (6D)  & 6   & 3186.8                                & 2918.83                                            & 6.3                        & 0.10 & 11.2 & 10.40 \\
    Peruffo (8D)  & 8   & 642.9                                 & 1238.83                                            & 3.3                        & 0.67 & 17.4 & 7.41  \\
    \midrule
    Darboux       & 2   & 1027.5                                & 708.85                                             & \multicolumn{2}{c}{--}     & 8.30 & 7.44         \\
    Polynomial    & 2   & 234.6                                 & 66.18                                              & 1.0                        & 0.00 & 2.1  & 1.52  \\
    Lyapunov      & 3   & 39.7                                  & 33.75                                              & 1.0                        & 0.00 & 1.3  & 0.48  \\
    \midrule
    Exponential   & 2   & 1799.0                                & 1270.11                                            & \multicolumn{2}{c}{--}     & 1.1  & 0.32         \\
    Ratschan (3D) & 3   & 65.4                                  & 5.02                                               & \multicolumn{2}{c}{--}     & 1.0  & 0.00         \\
    Ratschan (5D) & 5   & 49.9                                  & 19.68                                              & \multicolumn{2}{c}{--}     & 1.0  & 0.00         \\
    Ratschan (7D) & 7   & 284.3                                 & 428.92                                             & \multicolumn{2}{c}{--}     & 1.0  & 0.00         \\
    Ratschan (9D) & 9   & 399.1                                 & 553.93                                             & \multicolumn{2}{c}{--}     & 1.0  & 0.00         \\
    \bottomrule
  \end{tabular}
\end{table}

Recall that our approach does not require a separate verification step~(\cref{fig:teaser}): once training converges, the barrier certificate is valid by construction of the set-based loss function~(\cref{thm:soundness-loss}).
This paradigm shift makes a direct comparison of training epochs with related work ill-defined; nevertheless, \cref{tab:training_epochs} reports our training epochs and the number of training-verification cycles required by SynNBC and FOSSIL.

For most benchmarks, our approach terminates after only a few training epochs.
In contrast, FOSSIL often requires many training-verification cycles, each of which can itself involve up to 1000 training epochs~\cite{fossil2}.
This insight confirms the verification speed-up of our approach observed in \cref{tab:main-results}.

\begin{table*}
  \caption{Ablation Studies: Comparing required epochs and runtime for different neural networks and zero-level set computations. The results are averaged over 10 random neural network initializations.}
  \label{tab:ablationstudies}
  \fontsize{6.5}{7}\selectfont
  \begin{tabular}{c c R@{$\pm$}L R@{$\pm$}L R@{$\pm$}L R@{$\pm$}L R@{$\pm$}L R@{$\pm$}L R@{$\pm$}L R@{$\pm$}L R@{$\pm$}L}
    \toprule
                                 &          & \multicolumn{6}{c}{Peruffo (4D)} & \multicolumn{6}{c}{Lyapunov} & \multicolumn{6}{c}{Exponential}                                                                                                                                                                                                                                                                                                                                               \\
    \cmidrule(lr){3-8}  \cmidrule(lr){9-14}  \cmidrule(lr){15-20}
    $\iota$-$s$-$s_{\text{dim}}$ &          & \multicolumn{2}{c}{N1}           & \multicolumn{2}{c}{N2}       & \multicolumn{2}{c}{N3}           & \multicolumn{2}{c}{N1}    & \multicolumn{2}{c}{N2}     & \multicolumn{2}{c}{N3}    & \multicolumn{2}{c}{N1}   & \multicolumn{2}{c}{N2}    & \multicolumn{2}{c}{N3}                                                                                                                                                                          \\
    \midrule
    \multirow{3}{0.5cm}{1-8-n}   & Epochs   & 25.1                             & 5                            & 3801.8                           & 1180                      & 5795.3                     & 1723                      & 2603.0                   & 3534                      & \multicolumn{2}{c}{--}    & \multicolumn{2}{c}{--} & \multicolumn{2}{c}{--} & \multicolumn{2}{c}{--} & \multicolumn{2}{c}{--}                                                                   \\
                                 & Time [s] & 0.6                              & 0.02                         & 47.3                             & 14.61                     & 98.6                       & 22.03                     & 5.9                      & 7.31                      & \multicolumn{2}{c}{--}    & \multicolumn{2}{c}{--} & \multicolumn{2}{c}{--} & \multicolumn{2}{c}{--} & \multicolumn{2}{c}{--}                                                                   \\
                                 & Success  & \multicolumn{2}{c}{100\%}        & \multicolumn{2}{c}{40\%}     & \multicolumn{2}{c}{40\%}         & \multicolumn{2}{c}{ 20\%} & \multicolumn{2}{c}{ 0\%}   & \multicolumn{2}{c}{ 0\%}  & \multicolumn{2}{c}{ 0\%} & \multicolumn{2}{c}{ 0\%}  & \multicolumn{2}{c}{ 0\%}                                                                                                                                                                        \\

    \midrule
    \multirow{3}{0.5cm}{2-4-n}   & Epochs   & 28.0                             & 8                            & 3432.3                           & 2125                      & 3206.7                     & 1430                      & 957.8                    & 1460                      & 410.1                     & 466                    & 1027.3                 & 1176                   & \multicolumn{2}{c}{--} & \multicolumn{2}{c}{--} & \multicolumn{2}{c}{--}                 \\

                                 & Time [s] & 0.9                              & 0.16                         & 231.8                            & 155.13                    & 309.0                      & 172.08                    & 4.5                      & 5.94                      & 4.9                       & 4.81                   & 14.5                   & 16.13                  & \multicolumn{2}{c}{--} & \multicolumn{2}{c}{--} & \multicolumn{2}{c}{--}                 \\

                                 & Success  & \multicolumn{2}{c}{ 100\%}       & \multicolumn{2}{c}{ 60\%}    & \multicolumn{2}{c}{ 30\%}        & \multicolumn{2}{c}{ 60\%} & \multicolumn{2}{c}{ 80\%}  & \multicolumn{2}{c}{ 70\%} & \multicolumn{2}{c}{ 0\%} & \multicolumn{2}{c}{ 0\%}  & \multicolumn{2}{c}{ 0\%}                                                                                                                                                                        \\

    \midrule
    \multirow{3}{0.5cm}{2-8-n}   & Epochs   & 26.8                             & 9                            & 4798.2                           & 2157                      & 7531.0                     & 2246                      & 1317.0                   & 2911                      & 65.9                      & 45                     & 1869.8                 & 3447                   & \multicolumn{2}{c}{--} & 1753.3                 & 184                    & 1537.0 & 576  \\

                                 & Time [s] & 0.6                              & 0.04                         & 48.0                             & 26.65                     & 145.8                      & 92.39                     & 5.1                      & 10.01                     & 1.2                       & 0.43                   & 19.1                   & 27.20                  & \multicolumn{2}{c}{--} & 8.8                    & 0.95                   & 11.7   & 3.66 \\

                                 & Success  & \multicolumn{2}{c}{ 100\%}       & \multicolumn{2}{c}{ 90\%}    & \multicolumn{2}{c}{ 30\%}        & \multicolumn{2}{c}{ 50\%} & \multicolumn{2}{c}{ 100\%} & \multicolumn{2}{c}{ 80\%} & \multicolumn{2}{c}{ 0\%} & \multicolumn{2}{c}{ 60\%} & \multicolumn{2}{c}{ 90\%}                                                                                                                                                                       \\

    \midrule
    \multirow{3}{0.5cm}{4-2-n}   & Epochs   & 30.6                             & 14                           & 4378.8                           & 2357                      & 5417.8                     & 1375                      & 893.0                    & 1964                      & 397.7                     & 533                    & 1940.5                 & 3446                   & \multicolumn{2}{c}{--} & 2034.3                 & 1361                   & 1945.8 & 1222 \\
                                 & Time [s] & 1.3                              & 0.45                         & 361.8                            & 281.25                    & 439.6                      & 146.77                    & 7.2                      & 14.59                     & 5.6                       & 6.47                   & 32.2                   & 56.00                  & \multicolumn{2}{c}{--} & 12.0                   & 7.73                   & 14.3   & 8.65 \\
                                 & Success  & \multicolumn{2}{c}{ 100\%}       & \multicolumn{2}{c}{ 80\%}    & \multicolumn{2}{c}{ 50\%}        & \multicolumn{2}{c}{ 50\%} & \multicolumn{2}{c}{ 100\%} & \multicolumn{2}{c}{ 40\%} & \multicolumn{2}{c}{ 0\%} & \multicolumn{2}{c}{ 60\%} & \multicolumn{2}{c}{ 40\%}                                                                                                                                                                       \\
    \bottomrule
  \end{tabular}
\end{table*}

\subsection{Ablation Studies}

Subsequently, we conduct an ablation study to analyze the impact of
\begin{enumerate*}[label=(\roman*)]
  \item the neural network architecture and
  \item the parameters for the zero-level set computation.
\end{enumerate*}
In \cref{tab:ablationstudies}, we compare the number of epochs and the time in seconds required for the verification of three benchmarks with linear, polynomial, and nonpolynomial dynamics as well as the verification success.

\subsubsection{Experimental Setup} We consider three neural network architectures:
\begin{enumerate*}
  \item[(N1)] a single linear layer network with ten neurons ($\numLayers=1$),
  \item[(N2)] a network consisting of two linear and one nonlinear layer with eight hidden neurons ($\numLayers = 3$), and
  \item[(N3)] a deeper network having three linear and two nonlinear layers with five hidden neurons ($\numLayers = 5$).
\end{enumerate*}
In all architectures, the number of input neurons corresponds to the dimension of the system, and each neural network has a single output neuron.
For each of the neural networks, we investigate four different choices of parameters for the computation of the zero-level set enclosure in ten runs. The number of iterations $\iota$, splits $s$, and the number of dimensions along which we split $s_{\text{dim}}$ for the computation of the enclosure $\XzeroUnion$, are indicated in the first column of \cref{tab:ablationstudies}.

\subsubsection{Impact of the Network Architecture}
The results show that no architecture is the optimal choice across all benchmarks.
While the linear Peruffo~(4D) benchmark can be verified with the simple network N1 consistently in all configurations, the nonpolynomial Exponential benchmark cannot be verified with N1 within $10{,}000$ epochs.
This indicates, that the expressivity of the shallow network is not sufficient to verify complex systems.
In general, the required network complexity depends more on the dynamics of the system rather than the dimensionality.
For polynomial and nonpolynomial dynamics, deeper neural networks are required.
In particular, N2 performs best for the polynomial Lyapunov benchmark while N3 is required to verify the nonpolynomial Exponential benchmark.
These findings are consistent with the selection of networks that achieved the fastest and $100\%$ successful verification runs, as reported in \cref{tab:main-results}. In general, the runtime increases for networks with more layers, where N2 poses an exception in the case of the Lyapunov benchmark.

\subsubsection{Impact of the Zero-Level Set Enclosure}
The zero-level set computation has a negligible effect on the performance of the N1 architecture but significantly influences N2 and N3.
Systems with more complex dynamics tend to require a more elaborate zero-level set to satisfy property \thirdprop{}.
Therefore, more sets are necessary to ensure a tight enclosure.
Accordingly, configurations with larger $\iota$ and $s$ tend to improve the verification success for more complex dynamics.
This is visible in \cref{tab:ablationstudies}, as the bottom two rows have a higher number of splits or iterations, thereby enabling a more refined zero-level set enclosure.
Notably, a small number of epochs does not allow any conclusion to be drawn about the duration of the verification, as the time per epoch depends on both the number of sets in the zero-level set enclosure and the size of the neural network.

\subsubsection{Summary}
Overall, the ablation study highlights a fundamental trade-off.
While the verification with the simple N1 network is fast and reliable for the linear benchmark and has a moderate success rate for the polynomial benchmark, it fails for the nonpolynomial benchmark.
In contrast, deeper architectures N2 and N3 provide the expressive power required for polynomial and nonpolynomial systems, at the cost of increased computational effort.
Moreover, the choice of enclosure parameters becomes increasingly important as the complexity of the network and, especially, the system dynamics grows.
\section{Related Work}
The related work can be broadly categorized into three main research directions: traditional methods for generating barrier certificates through optimization, approaches using neural networks for the synthesis, and certified training techniques that integrate formal guarantees into the learning process.

\subsubsection{Barrier Certificate Generation}
The generation of barrier certificates can be formulated as an optimization problem~\cite{prajna2006}. There exist multiple approaches to simplify and generalize the synthesis of barrier certificates, like exponential conditions~\cite{kong2014}, compositional conditions~\cite{sloth2012}, Darboux polynomials~\cite{zeng2016}, and the combination of two functions to represent one barrier certificate~\cite{dai2017}. Other techniques employ interval analysis~\cite{bouissou2014} or simulation data~\cite{kapinski2014} to generate candidate barrier certificates. All of these approaches are unsound as they are based on iterative and numerical computation~\cite{peruffo2021} and cannot guarantee safety in general. Therefore, to guarantee the safety of the system, the generated candidate barrier certificates have to be formally verified in a subsequent step, e.g., using satisfiability modulo theories (SMT) solvers~\cite{kapinski2014}. 
PRoTECT \cite{Wooding2025} enables simultaneous verification of multiple candidate barrier certificates using parallelization.
These ideas have also been extended to stochastic systems~\cite{prajnaStoch2004} and to control barrier functions (CBFs)~\cite{ames2019,wieland2007}, which are used for real-time safe control synthesis rather than offline safety verification.

\subsubsection{Synthesis Using Neural Networks}
Neural networks are suitable for barrier certificate synthesis because, by the universal approximation theorem~\cite{cybenko1989,hornik1989,leshno1993}, they can represent any continuous function, in particular valid barrier certificates.
The general pipeline for synthesizing neural barrier certificates~\cite{zhao2020} follows two phases: first, a neural network is trained to satisfy the barrier certificate properties on samples of the state space, second, to ensure the validity of the barrier certificate across the entire state space, formal verification is required.

An instance of this pipeline is the counterexample-guided inductive synthesis (CEGIS)~\cite{peruffo2021}, which uses an iterative cycle between a learner and a verifier:
The learner updates the neural network to satisfy the barrier conditions, while the verifier (e.g., an SMT solver) either certifies the candidate or returns counterexamples that guide the next training iteration.
FOSSIL~\cite{fossil1,fossil2} implements this approach as a software tool.
There are different variants of CEGIS that mainly use different verifiers, e.g., cylindrical algebraic decomposition solvers~\cite{ding2022} that can produce multiple counterexamples, or sum-of-squares programming (SynNBC~\cite{zhao2023} and SynHBC~\cite{zhao2024SynHBC}).
The cycle can be augmented with a distillation phase between training and verification, which reduces the complexity of the neural network to simplify the subsequent verification task~\cite{ma2025}.
SEEV~\cite{zhang2024seev} takes a geometric approach by exploiting the piecewise-linear activation regions of ReLU neural networks.

A closely related line of work targets neural Lyapunov functions rather than barrier certificates, incorporating the required properties into the training loss~\cite{chang2019,dawson2023}.
A comprehensive survey of learning-based certificate synthesis is provided in~\cite{dawson2023}.

\subsubsection{Certified Training via Set-Based Propagation}
Formal verification of neural network properties using abstract interpretation~\cite{gehr2018,singh2019} and SMT solvers~\cite{katz2017reluplex} has motivated the idea of incorporating set-based computing directly into training~\cite{mirman2018,gowal2019,koller2025shadows}, which mainly focuses on input-output robustness of classifiers; instead, we propagate sets to dynamically verify neural network properties during training, thereby eliminating the training-verification cycle entirely.

\section{Conclusions}
We propose a novel set-based training approach to synthesize neural barrier certificates.
Prior work follows an iterative training-verification cycle: train a candidate barrier certificate on samples of the state space, verify its validity with an external solver, and repeat until success.
We replace this costly cycle with a single set-based training step, where verification is integrated directly into the training loss function.
Our sound set-based loss encodes all three barrier certificate properties via conservative set propagation over zonotopes; when the loss reaches zero, the barrier certificate is valid by construction and no separate verification step is needed.
All operations reduce to efficient matrix multiplications, yielding polynomial complexity (\cref{sec:app:complexity}) that lets our set-based training scale to high-dimensional systems.
Moreover, the approach is largely agnostic to the type of system, 
such that it can also naturally handle complex nonlinear dynamics.
Our experiments demonstrate that our approach outperforms prior works on most benchmarks, particularly on high-dimensional and nonpolynomial systems.
Thus, our set-based synthesis of barrier certificates can efficiently generate barrier certificates without requiring subsequent verification.

\appendix 
\crefalias{section}{appendix}
\crefalias{subsection}{appendix}
\subsection{Computational Complexity}\label{sec:app:complexity}
In this section, we analyse the computational complexity of a single training iteration in~\cref{alg:learn-nbc} and discuss the convergence of the training.

Before stating the result, we introduce some further notation.
During a set-based forward propagation~(\cref{def:setbasedForwardProp}), the propagated zonotopes have at most $\numGens\leq\numNeurons_0+\sum_{k\in[\numLayers]}\numNeurons_k$ generators, and we write $\numNeurons_\text{max}\coloneqq\max_{k\in\{0\}\cup[\numLayers]}\numNeurons_k$
for the maximum layer width; in particular, the state dimension satisfies $n=\numNeurons_0\leq\numNeurons_\text{max}$.
Enclosing the output of the neural network takes time $\bigO(\numNeurons_\text{max}^2\,\numGens\,\numLayers)$~\cite[Prop.~16]{koller2025}.
Finally, let $\tau_f$ denote the number of elementary operations needed to evaluate~$f$, i.e., a single interval-arithmetic enclosure of~$f$ takes time~$\bigO(\tau_f)$.

A training iteration~(\cref{alg:learn-nbc}) is linear in the number of zero-level sets~$\numZsets$ and in~$\tau_f$, and quadratic in the layer width, and in particular in the state dimension~$n$.

\begin{proposition}[Time Complexity of an Iteration in~\cref{alg:learn-nbc}]\label{prop:training_complexity}
  An iteration in~\cref{alg:learn-nbc} has time complexity $\bigO(\numZsets\,\numNeurons_\text{max}^2\,\numGens\,\numLayers + \numZsets\,\tau_f)$ w.r.t. $\numNeurons_\text{max}$, $\numGens$, the number of layers $\numLayers$, the number of zero-level sets~$\numZsets$, and $\tau_f$.
  \begin{proof}
    Computing~$\setLossI$ and~$\setLossII$ each require
    \begin{enumerate*}[label=(\roman*)]
      \item a set-based forward propagation~($\bigO(\numNeurons_\text{max}^2\,\numGens\,\numLayers)$) and
      \item computing the interval bounds of the output zonotope~($\bigO(\numGens)$).
    \end{enumerate*}
    Hence, computing loss~$\setLossI$ and~$\setLossII$ takes a total time of~$\bigO(\numNeurons_\text{max}^2\,\numGens\,\numLayers)$.

    Moreover, $\setLossIII$ is computed for each zero-level set, which requires
    \begin{enumerate*}[label=(\roman*)]
      \item enclosing the set of gradients~($\bigO(\numNeurons_\text{max}^2\,\numGens\,\numLayers)$),
      \item enclosing the Lie derivative~($\bigO(\numGens\,n + \tau_f)$), and
      \item enclosing the zonotope product~($\bigO(\numGens\,n^2)$).
    \end{enumerate*}
    Hence, in total, computing~$\setLossIII$ takes time~$\bigO(\numZsets\,(\numNeurons_\text{max}^2\,\numGens\,\numLayers + \numGens\,n + \tau_f + \numGens\,n^2)) = \bigO(\numZsets\,\numNeurons_\text{max}^2\,\numGens\,\numLayers + \numZsets\tau_f)$; since $n\leq\numNeurons_\text{max}$ and $n\leq\numGens$, the terms $\numGens\,n^2$ and $\numGens\,n$ are absorbed into $\numNeurons_\text{max}^2\numGens\numLayers$.

    Moreover, aggregating the losses and updating the parameters using backpropagation takes time~$\bigO(\numNeurons_\text{max}^2\,\numGens\,\numLayers)$~\cite[Sec.~4.4]{koller2025}.

    Thus, in total, an iteration in~\cref{alg:learn-nbc} takes time~$\bigO(\numZsets\,\numNeurons_\text{max}^2\,\numGens\,\numLayers + \numZsets\tau_f)$.
  \end{proof}
\end{proposition}

\subsection{Proofs}
\label{app:proofs}

In this section, we provide all missing proofs.

\renewcommand{\thetheorem}{\ref{thm:soundness-loss}}
\begin{proof}\textit{(of \cref{thm:soundness-loss})}
  We assume
  \begin{align}\label{eq:thm:soundness-loss:proof:assms}
    \setLoss(\Yu,\Yinit,\YzeroUnion) = 0\text{,}
  \end{align}
  and we show the safety of the system in~\eqref{eq:cont_sys}.
  Applying \cref{thm:barrier}, we show that the neural network~$\BarrierNN$ satisfies the properties \firstprop{} -- \thirdprop{}.
  Moreover, with assumption~\eqref{eq:thm:soundness-loss:proof:assms} and definition~\eqref{eq:total-loss}, we have $
    \setLossI(\Yu) + \setLossII(\Yinit) + \setLossIII(\YzeroUnion) = 0$, and by the nonnegativity of the individual losses  $\setLossI$, $\setLossII$, and $\setLossIII$, we have
  \begin{align*}
    \setLossI(\Yu)  \overset{\eqref{eq:lossI}}             & {=} 0\text{,} &
    \setLossII(\Yinit)  \overset{\eqref{eq:lossII}}        & {=} 0\text{,} &
    \setLossIII(\YzeroUnion)  \overset{\eqref{eq:lossIII}} & {=} 0\text{.}
  \end{align*}
  From $\setLossI = 0$ with~\eqref{eq:lossI}, we can deduce
  \begin{multline*}
    \max(0,-\buLower + \epsilon) = 0 \iff -\buLower + \epsilon \leq 0 \\
    \iff \buLower > 0 \implies \forall x\in\Xu\colon \BarrierNN(x) > 0\text{.}
  \end{multline*}
  Thus, the neural network $\BarrierNN$ satisfies property~\firstprop{}.
  Moreover, from $\setLossII = 0$ with \eqref{eq:lossII}, we can deduce
  \begin{align*}
    \max(0,\binitUpper) = 0 \iff \binitUpper \leq 0 \implies \forall x\in\Xinit\colon \BarrierNN(x) \leq 0\text{.}
  \end{align*}
  Thus, the neural network $\BarrierNN$ satisfies property~\secondprop{}.
  Moreover, from $\setLossIII = 0$, we can deduce for each $i\in[\numZsets]$,
  \begin{multline*}
    \max(0,\bzeroUpperSub{i} + \epsilon) = 0 \iff \bzeroUpperSub{i} + \epsilon \leq 0 \\ \implies \forall x\in\Xzero\colon \lie_f\BarrierNN(x) < 0\text{.}
  \end{multline*}
  Thus, the neural network $\BarrierNN$ satisfies property~\thirdprop{}.
\end{proof}

\renewcommand{\theproposition}{\ref{prop:input_set_refinement}}
\begin{proof}\textit{(of \cref{prop:input_set_refinement})}
  This follows directly from~\cite[Prop.~2]{koller2025shadows}, where we want to enclose the zero-level set, thus,
  \begin{equation*}
    \mathcal{U}=\{0\} = \{y\in\R\mid [1\ -1]^\top y \leq \0 \}\text{.}
  \end{equation*}
\end{proof}

\renewcommand{\theproposition}{\ref{prop:enclosing-setbackprop}}

\begin{proof}\textit{(of \cref{prop:enclosing-setbackprop})}
  We prove the gradient enclosure by induction on $k=\numLayers,\dots,1$.
  Let $\nnHiddenSet_{0} = \XzeroSub{i}$, and $\nnHiddenSet_{1}, \dots, \nnHiddenSet_{\numLayers}$ be the enclosures of output sets of the hidden layers~(\cref{def:setbasedForwardProp}).
  If $k=\numLayers$, all gradients are $\nabla_{\nnHidden_{\numLayers}}\BarrierNN(\nnInput) \overset{\text{\cref{def:setbasedForwardProp}}}{=} \nabla_{\nnHidden_{\numLayers}}\nnHidden_{\numLayers} = \ones$, hence
  \begin{align}
    \nnGradSet_\numLayers = \shortZ{1}{\0} \supseteq \{\nabla_{\nnHidden_{\numLayers}}\nnHidden_{\numLayers}\mid \nnHidden_{\numLayers}\in\nnHiddenSet_{\numLayers}\}\text{.}
  \end{align}
  For $k<\numLayers$, we assume
  \begin{align}\label{eq:proof:grad_enc:assms}
    \nnGradSet_{k} \supseteq \{\nabla_{\nnHidden_{k}}\BarrierNN(\nnInput)\mid \nnHidden_{k}\in\nnHiddenSet_{k}\}\text{,}
  \end{align}
  and we show $\nnGradSet_{k-1} \supseteq \{\nabla_{\nnHidden_{k-1}}\BarrierNN(\nnInput)\mid \nnHidden_{k-1}\in\nnHiddenSet_{k-1}\}$.
  We split cases on the type of the $k$-th layer and simplify the terms.
  \begin{enumerate}[label=Case (\roman*)., wide=0pt, font=\itshape]
    \item The $k$-th layer is linear. Hence, we have~\cite[Sec.~5.3]{bishop2006}
          \begin{align}\label{eq:proof:grad_enc:linear_grad}
            \nabla_{\nnHidden_{k-1}}\BarrierNN(\nnInput) = W_k^\top\,\nabla_{\nnHidden_{k}}\BarrierNN(\nnInput)\text{.}
          \end{align}
          Thus,
          \begin{align*}
            \nnGradSet_{k-1}                                       & = W_k^\top\,\nnGradSet_{k} \overset{\text{\eqref{eq:proof:grad_enc:assms}}}{\supseteq} W_k^\top\,\{\nabla_{\nnHidden_{k}}\BarrierNN(\nnInput)\mid \nnHidden_{k}\in\nnHiddenSet_{k}\} \\
                                                                   & = \{W_k^\top\,\nabla_{\nnHidden_{k}}\BarrierNN(\nnInput)\mid \nnHidden_{k}\in\nnHiddenSet_{k}\}                                                                                      \\
            \overset{\text{\eqref{eq:proof:grad_enc:linear_grad}}} & {=} \{\nabla_{\nnHidden_{k-1}}\BarrierNN(\nnInput)\mid \nnHidden_{k-1}\in\nnHiddenSet_{k-1}\}\text{.}
          \end{align*}
    \item The $k$-th layer is nonlinear. We enclose the derivative of the activation function, i.e.,
          \begin{align}
            \opEnclose{\nnActFun_k'}{\nnHiddenSet_{k-1}}\supseteq \{\nabla_{\nnHidden_{k-1}}\nnActFun_k'(\nnHidden_{k-1})\mid \nnHidden_{k-1}\in\nnHiddenSet_{k-1}\}\text{.}
          \end{align}
          Moreover, we have~\cite[Sec.~5.3]{bishop2006} for each $i\in[\numNeurons_{k}]$
          \begin{align}\label{eq:proof:grad_enc:nonlinear_grad}
            \nabla_{\nnHidden_{k-1(i)}}\BarrierNN(\nnInput) = \nnActFun_k'(\nnHidden_{k-1(i)})\,\nabla_{\nnHidden_{k(i)}}\BarrierNN(\nnInput)\text{.}
          \end{align}
          Thus,
          \begin{align*}
            \nnGradSet_{k-1}                                                          & = \opEnclose{\nnActFun_k'}{\nnHiddenSet_{k-1}}\cdot\nnGradSet_k                                                                                  \\
            \overset{\text{\eqref{eq:proof:grad_enc:assms}}}                          & {\supseteq} \opEnclose{\nnActFun_k'}{\nnHiddenSet_{k-1}}\cdot\{\nabla_{\nnHidden_{k}}\BarrierNN(\nnInput)\mid \nnHidden_{k}\in\nnHiddenSet_{k}\} \\
            \overset{\eqref{eq:zonoproduct}}                                          & {\supseteq} \{\nnActFun_k'(\nnHidden_{k-1(i)})\,\nabla_{\nnHidden_{k}}\BarrierNN(\nnInput)\mid \nnHidden_{k}\in\nnHiddenSet_{k}\}                                        \\
            \overset{\text{\eqref{eq:proof:grad_enc:nonlinear_grad}}}                 & {=} \{\nabla_{\nnHidden_{k-1}}\BarrierNN(\nnInput)\mid \nnHidden_{k-1}\in\nnHiddenSet_{k-1}\}\text{.}
          \end{align*}
  \end{enumerate}
\end{proof}

\subsection{Details on Experiments and Benchmarks}
\label{app:details-experiment-benchmark}

\subsubsection{Experimental Details}
\begin{table}
  \centering
  \newcolumntype{.}{D{.}{.}{1.3}}
  \newcolumntype{,}{D{.}{.}{2}}
  \caption{Parameter choice for each dynamic system.}
  \label{tab:parameterChoice}
  \begin{tabular}{l l c c , . c c}
    \toprule
    Dynamic                 & Benchmark     & $n$ & $\numLayers$ & \multicolumn{1}{c}{$\numNeurons_k$} & \multicolumn{1}{c}{$\eta$} & $\beta_1$ & $\iota$-$s$-$s_{\text{dim}}$ \\
    \midrule
    \multirow{5}{*}{Linear} & Three Sets    & 2   & 1            & -                                   & 0.1                        & 0.3       & 1-3-n                        \\
                            & Two Barriers  & 2   & 3            & 10                                  & 0.1                        & 0.3       & 2-4-n                        \\
                            & Peruffo~(4D)  & 4   & 1            & -                                   & 0.1                        & 0.9       & 1-8-n                        \\
                            & Peruffo~(6D)  & 6   & 1            & -                                   & 0.1                        & 0.3       & 2-4-2                        \\
                            & Peruffo~(8D)  & 8   & 1            & -                                   & 0.1                        & 0.3       & 2-4-2                        \\
    \midrule
    \multirow{3}{*}{Poly.}  & Darboux       & 2   & 3            & 8                                   & 0.1                        & 0.3       & 2-9-n                        \\
                            & Polynomial    & 2   & 3            & 8                                   & 0.01                       & 0.3       & 2-8-n                        \\
                            & Lyapunov      & 3   & 3            & 8                                   & 0.1                        & 0.9       & 1-32-n                       \\
    \midrule
    \multirow{5}{*}{Nonpoly.}
                            & Exponential   & 2   & 5            & 5                                   & 0.001                      & 0.3       & 2-10-n                       \\
                            & Ratschan~(3D) & 3   & 1            & -                                   & 0.1                        & 0.3       & 1-3-n                        \\
                            & Ratschan~(5D) & 5   & 1            & -                                   & 0.1                        & 0.3       & 2-2-n                        \\
                            & Ratschan~(7D) & 7   & 1            & -                                   & 0.1                        & 0.3       & 2-2-2                        \\
                            & Ratschan~(9D) & 9   & 1            & -                                   & 0.1                        & 0.3       & 2-4-2                        \\
    \bottomrule
  \end{tabular}
\end{table}

The experiments are conducted on an Apple MacBook Pro (13-inch, 2020) equipped with Apple's M1 chip and 16 GB of unified memory.
The FOSSIL experiments were run on a laptop with an Intel Core i7-13700H, since FOSSIL's Python dependencies do not install on Apple Silicon. This does not disadvantage FOSSIL, as the i7-13700H is a more recent and more powerful CPU than the M1.

To enclose the flow of nonlinear systems in \cref{sec:lie-derivative}, we use range bounding based on interval arithmetic~\cite[Sec. 2.3.3]{jaulin2001} to obtain bounds for $f(\XzeroSub{i})$.
We reduce the conservativeness of interval arithmetic by splitting the zero-level sets $\XzeroSub{i}$.
This is mainly done for computational efficiency, but one could also apply techniques similar to \cref{def:setbasedForwardProp} to obtain tighter enclosures.

If not stated otherwise, each experiment is repeated 10 times.
For the evaluation provided in \cref{tab:main-results}, we use the parameters in \cref{tab:parameterChoice}.
The neural networks consist of a specified number of linear layers and hidden neurons, where the input dimension is the system dimension and the output dimension is one. We use the Adam optimizer with learning rate $\eta$ and momentum parameter $\beta_1$.
Further, we pre-train each neural network for 10 epochs to fit a radial basis function centered around the initial set.
The zero-level set $\XzeroUnion$ computation is based on three criteria:
The number of iterations $\iota$, the number of splits $s$, and the number of dimensions along which we split $s_{\text{dim}}$.

For the ablation studies summarized in \cref{tab:ablationstudies}, we use $\eta = 0.1$, $\beta_1=0.9$ for the Peruffo~(4D) and the Lyapunov benchmark, and $\eta = 0.001$, $\beta_1 = 0.3$ for the Exponential benchmark. In general, the number of split dimensions of $\XzeroUnion$ is $n$. The zero-level set computation, consisting of two iterations and eight splits, has split dimension two for all three benchmarks.

\subsubsection{Benchmarks}

In this section, we provide the dynamics, initial sets, and unsafe sets for all benchmarks.

\paragraph{Three Sets}
\begin{equation}
  \begin{split}
    \begin{gathered}
      f(x) = \cmat{0 \\ x_{(2)}}\text{, }
      \X  = \shortI{0}{4}^2\text{, }
      \Xinit  = \shortI{\cmat{1.7 \\ 2.7}}{\cmat{2.3 \\ 3.3}}\text{,}
      \\
      \Xu = \shortI{0.7}{2.3}^2 \cup \shortI{\cmat{2.7 \\ 1.7}}{\cmat{3.3 \\ 2.3}}.
    \end{gathered}
  \end{split}
\end{equation}

\paragraph{Two Barriers}
\begin{equation}
  \begin{split}
    \begin{gathered}
      f(x) = \cmat{-x_{(1)} \\ -x_{(2)}}\text{,} \\
      \begin{aligned}
        \X     & = \shortI{0}{4}^2\text{,}                              &
        \Xinit & = \shortI{0.7}{1.3}^2 \cup \shortI{2.7}{3.3}^2\text{,}
      \end{aligned} \\
      \Xu = \shortI{\cmat{0.3 \\ 3.3}}{\cmat{0.7 \\ 3.7}} \cup \shortI{\cmat{3.3 \\ 0.3}}{\cmat{3.7 \\ 0.7}}.
    \end{gathered}
  \end{split}
\end{equation}

\paragraph{Peruffo ($n$D)~\cite{peruffo2021}}
\begin{equation}
  \begin{split}
    \begin{gathered}
      f_n(x) = \cmat{ x_{(1)} & \cdots & x_{(n)} & \sum_{i=1}^n w_{(i)}x_{(i)}}^\top\text{,} \\
      w \coloneqq -\cmat{576,\ 2400,\ 4180,\ 3980,\ 2273,\ 800,\ 170,\ 20}\text{,} \\
      \begin{aligned}
        \X     & = \shortI{-2.2}{2.2}^n\text{,} &
        \Xinit & = \shortI{0.9}{1.1}^n\text{,}  &
        \Xu    & = \shortI{-2.2}{-1.8}^n.
      \end{aligned}
    \end{gathered}
  \end{split}
\end{equation}

\paragraph{Darboux~\cite[System 8]{zeng2016}}
\begin{equation}
  \begin{split}
    \begin{gathered}
      f(x) = \cmat{x_{(2)} + 2x_{(1)}x_{(2)} \\ -x_{(1)} + 2x_{(1)}^2 - x_{(2)}^2}\text{,}\
      \X = \shortI{-2}{2}^2\text{,}  \\
      \begin{aligned}
        \Xinit & = \shortI{\cmat{0                     \\ 1}}{\cmat{1 \\ 2}}\text{,} &
        \Xu    & = \{x_{(1)}+x_{(2)}^2\leq 0\}\text{.}
      \end{aligned}
    \end{gathered}
  \end{split}
\end{equation}

In our case, we over-approximate the unsafe set with a zonotope:
\begin{equation*}
  \X_u = \shortZ[.]{\begin{bmatrix}
      -2 \\
      0
    \end{bmatrix}}{\begin{bmatrix}
      0    & 0.75   & 0.75  & 0.25  & 0.25 \\
      0.25 & -0.375 & 0.375 & -0.25 & 0.25
    \end{bmatrix}}
\end{equation*}

\paragraph{Polynomial~\cite{prajna2007} (intervals instead of ellipsoids)}
\begin{align}
  f(x)   & = \cmat{x_{(2)}   \\ -x_{(1)} + 1/3x_{(1)}^3- x_{(2)}}\text{, }
  \X = \shortI{\cmat{-3.5    \\ -2}}{\cmat{2 \\ 1}}\text{,} \nonumber \\
  \Xinit & = \shortI{\cmat{1 \\ -0.5}}{\cmat{2 \\ 0.5}}\text{, }
  \Xu = \shortI{-1.4}{-0.6}^2\text{.}
\end{align}

\paragraph{Lyapunov (based on \cite[Example 7]{ratschan2006})}
\begin{equation}
  \begin{split}
    \begin{gathered}
      f(x) = \cmat{-x_{(2)} \\ -x_{(3)} \\ -x_{(1)} - 2 x_{(2)}-x_{(3)} + x_{(1)}^3}\text{,}\
      \X = \shortI{-2}{2}^3\text{,} \\
      \begin{aligned}
        \Xinit & = \shortI{\cmat{-0.25 \\ -0.25 \\ -0.75}}{\cmat{0.75 \\ 0.75 \\ 0.25}}\text{,} &
        \Xu    & = \shortI{\cmat{1     \\ -2 \\ -2}}{\cmat{2 \\ -1 \\ -1}}\text{.}
      \end{aligned}
    \end{gathered}
  \end{split}
\end{equation}

\paragraph{Exponential (based on \cite[Example 1]{liu2015})}
\begin{equation}
  \begin{split}
    \begin{gathered}
      f(x) = \cmat{e^{-x_{(1)}} +x_{(2)} -1 \\ -\sin^2x_{(1)}}\text{,}\  \X = \shortI{-2}{2}^2\text{,} \\
      \begin{aligned}
        \Xinit & = \shortI{-0.9}{-0.1}^2\text{,} &
        \Xu    & = \shortI{0.4}{1}^2\text{.}
      \end{aligned}
    \end{gathered}
  \end{split}
\end{equation}

\paragraph{Ratschan ($n$D) (based on \cite[Example 5]{ratschan2017})}
\begin{equation}
  \begin{split}
    \begin{gathered}
      f_l(x) = \cmat{
      1 + \frac{1} {100} (\sum_{i\in [l]} (x_{(i+1)} + x_{(i+2)}))\\
      x_{(3)} \\
      -10 \sin x_{(2)}-x_{(2)} \\
      \cdots \\
      x_{(2l+1)} \\
      -10 \sin x_{(2l)}-x_{(2l)}
      }\text{,} \\
      \X = \shortI{-0.3}{0.3}^{2l+1}\text{, }
      \Xinit = \shortI{-0.3}{0}\times\shortI{-0.2}{0.3}^{2l}\text{,}
      \\
      \Xu = \shortI{-0.2}{-0.15}\times\shortI{-0.3}{-0.25}^{2l}\text{.}
    \end{gathered}
  \end{split}
\end{equation}

\section*{References}
\bibliographystyle{IEEEtran}
\bibliography{references}

\begin{IEEEbiography}[
    {\includegraphics[width=1in,height=1.25in,clip,keepaspectratio]{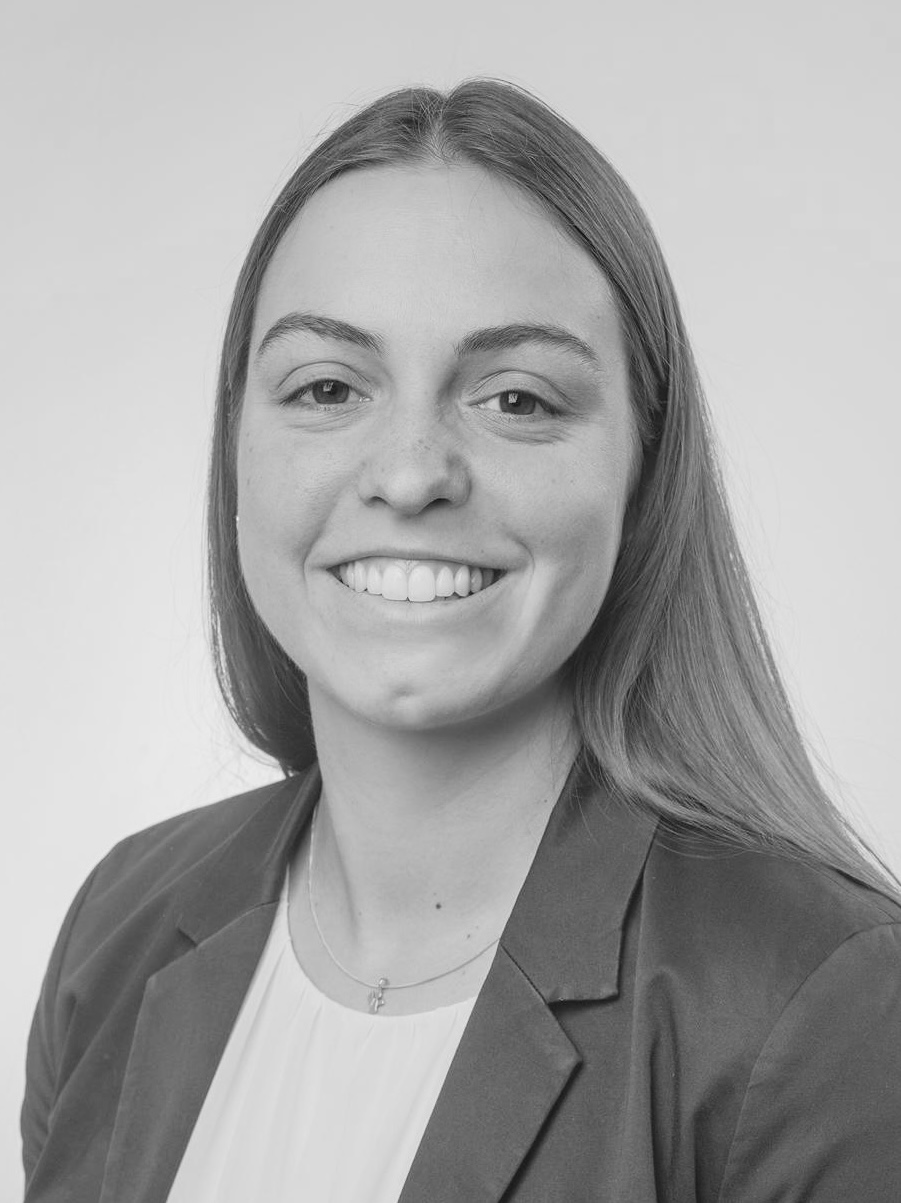}}
  ]
  {Miriam Kranzlmüller}
  received the B.Sc. degrees in mathematics and in computer science from Ludwig-Maximilians-Universität in Munich, Germany, in 2022, and the M.Sc. degree in computer science from the Technical University of Munich, Germany, in 2025. She is currently pursuing the Ph.D. degree at Ludwig-Maximilians-Universität in Munich. Her research interests include formal verification of neural networks as well as theoretical foundations for learning of spiking neural networks.
\end{IEEEbiography}

\begin{IEEEbiography}[
    {\includegraphics[width=1in,height=1.25in,clip,keepaspectratio]{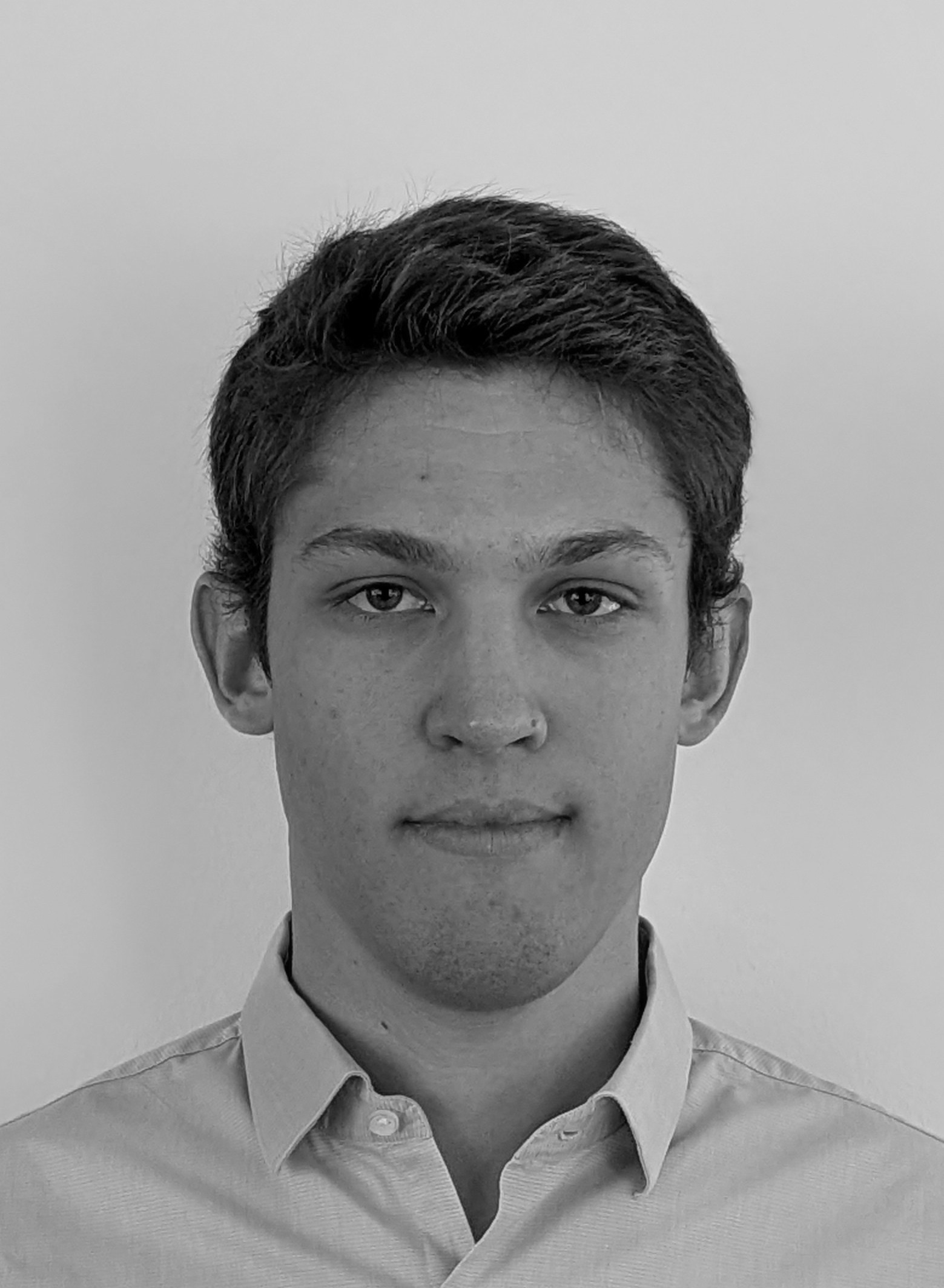}}
  ]
  {Lukas Koller} received the B.Sc. and M.Sc. degrees in computer science from the Technical University of Munich, Germany, in 2021 and 2023, respectively. He is currently pursuing the Ph.D. degree at the Technical University of Munich. His research interests include formal verification and robust training of neural networks using set-based methods.
\end{IEEEbiography}

\begin{IEEEbiography}[{\includegraphics[width=1in,height=1.25in,clip,keepaspectratio]{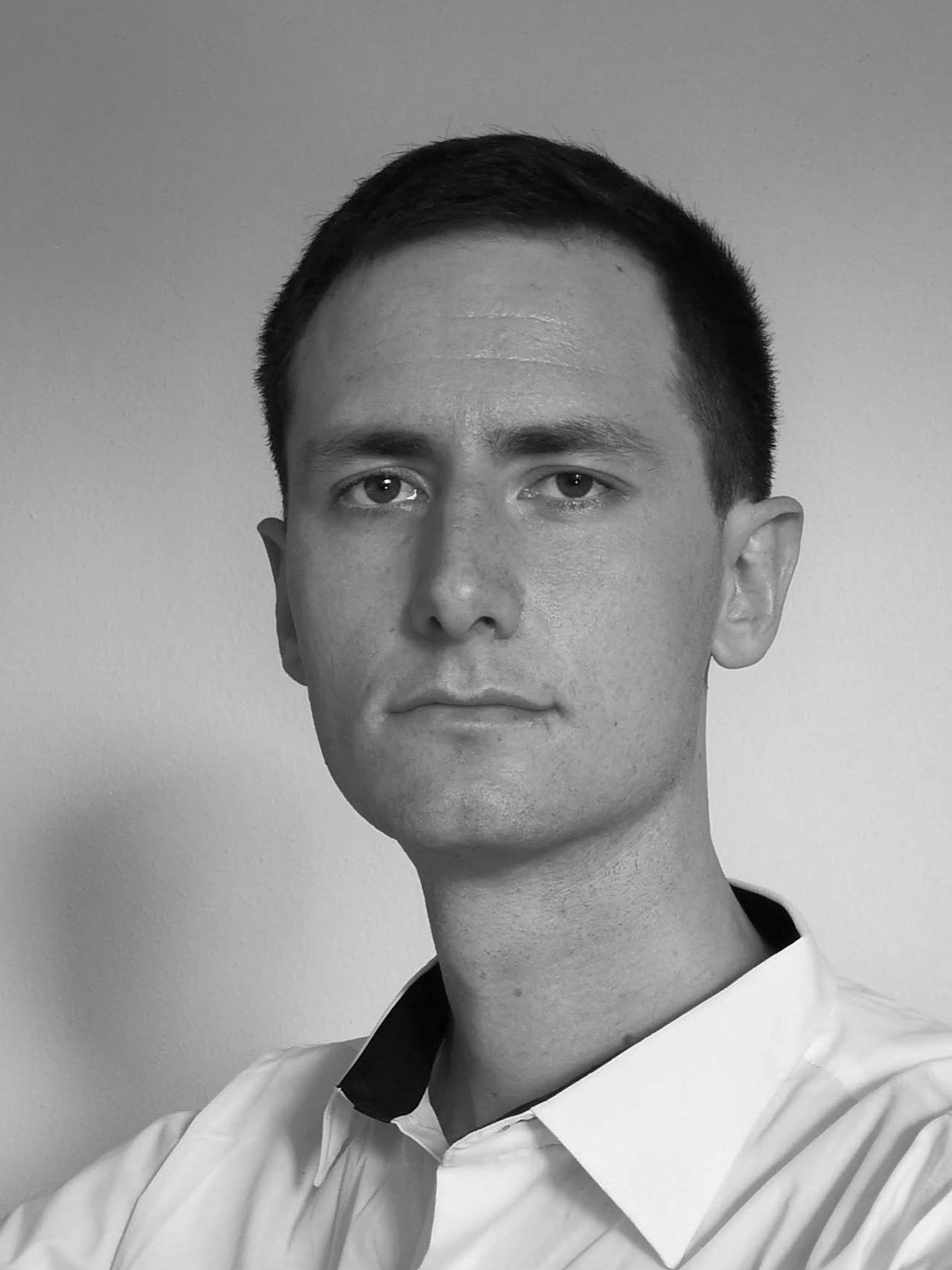}}]
  {Tobias Ladner} received the B.Sc. and M.Sc. degrees in computer science from the Technical University of Munich, Germany, in 2019 and 2021, respectively. He is currently pursuing the Ph.D. degree at the Technical University of Munich and was a visiting researcher at the University of California, Irvine. His research interests include the formal safety of artificial intelligence, with a focus on adversarial robustness properties and explainability, as well as their deployment in neural network control systems.
\end{IEEEbiography}

\begin{IEEEbiography}[{\includegraphics[width=1in,height=1.25in,clip,keepaspectratio]{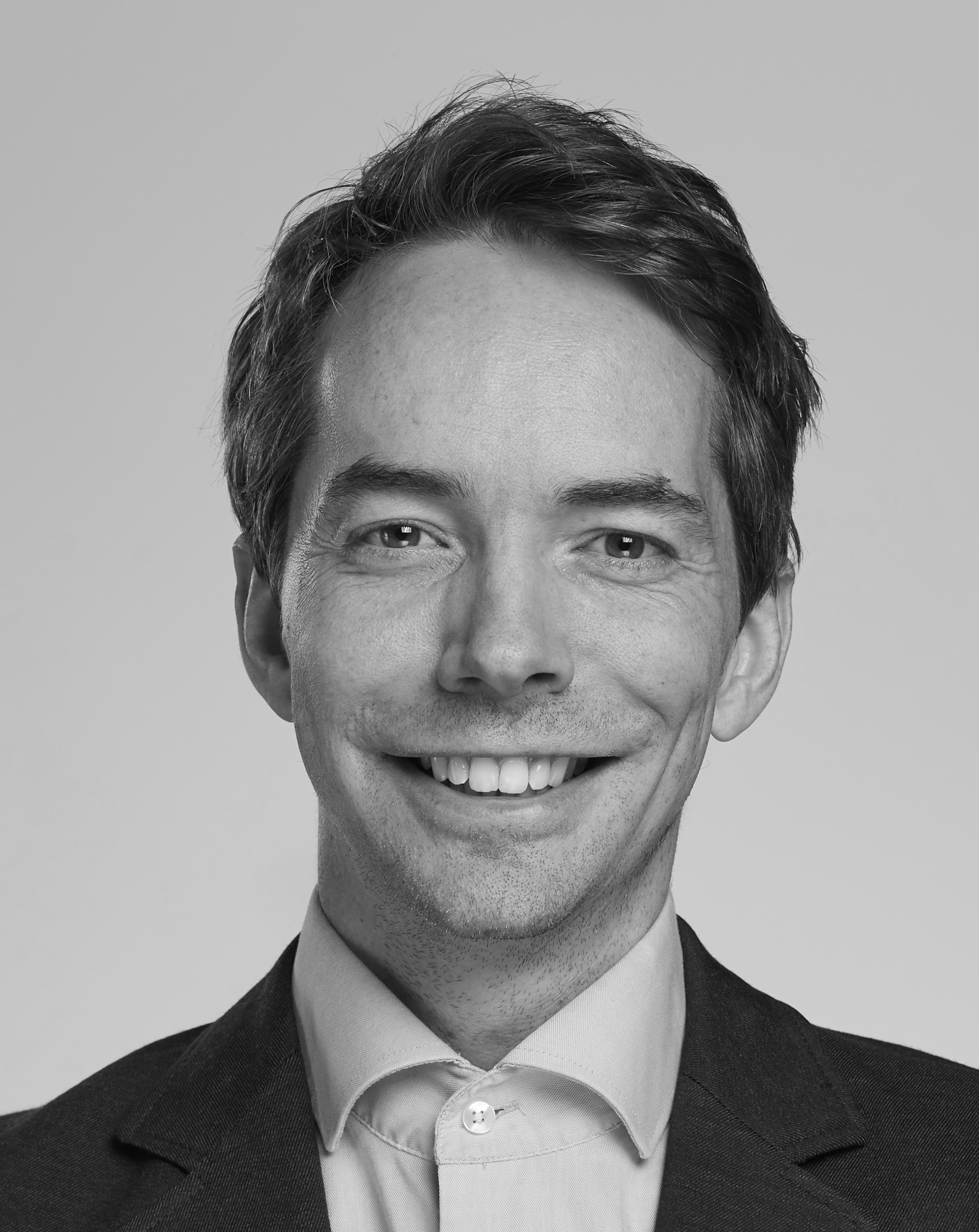}}]
  {Matthias Althoff}
  received the diploma in mechanical engineering and the Ph.D. degree in
  electrical engineering from the Technical University of Munich, Germany, in 2005 and
  2010, respectively.
  From 2010 to 2012, he was a Postdoctoral
  Researcher with Carnegie Mellon University,
  Pittsburgh, PA, USA, and from 2012 to 2013, an
  Assistant Professor with Technische Universität
  Ilmenau, Ilmenau, Germany. He is currently an
  Associate Professor in computer science with
  Technical University of Munich. His research interests include formal
  verification of continuous and hybrid systems, reachability analysis,
  planning algorithms, nonlinear control, automated vehicles, and power
  systems.
\end{IEEEbiography}

\end{document}